%% file: main.tex
\documentclass[journal,twoside,web]{ieeecolor}
\usepackage{generic}

\usepackage{amsmath,amssymb,amsfonts, mathrsfs}
\usepackage{dsfont}
\usepackage{matlab-prettifier}
\usepackage{algpseudocode}
\usepackage{algorithm}
\usepackage{array}
\usepackage{textcomp}
\usepackage{stfloats}
\usepackage{url}
\usepackage{verbatim}
\usepackage{cite}
\usepackage{tikz}
\usepackage{xcolor}
\usepackage{arydshln}
\usepackage{graphicx}
\usepackage{subcaption}
\usepackage{hyperref}
\hypersetup{hidelinks=true}
\usepackage{textcomp}
\def\BibTeX{{\rm B\kern-.05em{\sc i\kern-.025em b}\kern-.08em
    T\kern-.1667em\lower.7ex\hbox{E}\kern-.125emX}}
\newcommand{\mc}[1]{\mathcal{#1}}
\newcommand{\mr}[1]{\mathrm{#1}}
\newcommand{\ms}[1]{\mathscr{#1}}

\newcommand{\mb}[1]{\mathbb{#1}}

\newcommand{\sbullet}{{\vcenter{\hbox{\tiny$\bullet$}}}}

\newcommand{\Matrix}[3][1.25]{\ensuremath{
		{\arraycolsep= #1 pt \def\arraystretch{1}
			\left[ \begin{array}{#2} #3 \end{array} \right]}}}

\definecolor{bluefig}{RGB}{22,80,225}
\definecolor{orangefig}{RGB}{254,153,0}
\definecolor{redfig}{RGB}{151,21,0}
\definecolor{dgreenfig}{RGB}{126,203,72}
\definecolor{lgreenfig}{RGB}{201,233,179}

\newtheorem{definition}{Definition}
\newtheorem{theorem}{Theorem}

\newtheorem{remark}{Remark}

\newlength{\optlabelwidth}
\newlength{\optrhswidth}
\newlength{\optcolsep}

\begin{document}

\title{Robust Data-Driven Controller Synthesis from Sampled and Uncertain Frequency Response Estimates using Integral Quadratic Constraints}

\author{L. M. Spin, 
W. P. M. H. Heemels, 
N. van de Wouw, 
R. T\'oth 
 \thanks{ 
This work is co-financed by Holland High Tech, top sector High-Tech Systems and Materials, with a PPP innovation subsidy for public-private partnerships for research and development. 
Maurice Heemels has also received funding from the European Research Council under the Advanced ERC Grant Agreement PROACTHIS, no. 101055384.}
\thanks{L. M. Spin and R. T\'oth are with the Control Systems Group, N. van de Wouw is with the Dynamics and Control Group, W. P. M. H. Heemels is with the Control Systems Technology Group, Eindhoven University of Technology, The Netherlands. R. T\'oth is also with the Institute for Computer Science and Control, Budapest, Hungary. Corresponding author: L. M. Spin, email address: \texttt{l.m.spin@tue.nl}.}}


\maketitle

\begin{abstract}
In this paper, we develop a novel MIMO controller synthesis method that relies solely on frequency-domain data of discrete-time systems. The approach is based on an iterative design procedure using \emph{integral quadratic constraints} (IQCs). It generalizes existing data-driven methods that approximate semi-infinite IQC conditions by enforcing them on a finite set of measured frequency points. Such approximations do not provide guarantees on the continuous frequency axis and may fail to capture inter-frequency behavior. We address this limitation by formulating a convex optimization problem that preserves robustness guarantees over the entire frequency range. In particular, the proposed method explicitly accounts for inter-frequency behavior of both plant and controller, leading to certified robust stability and performance. The effectiveness and computational efficiency of the approach are demonstrated using experimentally obtained frequency response data from industrial-grade semiconductor equipment. The synthesized controller is implemented on the physical setup, and its performance is validated in real-time experiments.
\end{abstract}


\input{sections/Introduction.tex}
\input{sections/ProblemFormulation.tex}

\input{sections/SolutionApproach.tex}
\input{sections/MainResult.tex}
\input{sections/GeneratorConstruction.tex}

\input{sections/Wirebonder.tex}
\input{sections/Conclusion.tex}
\input{sections/Acknowledgement.tex}
\input{sections/AppendixA.tex}
\bibliographystyle{IEEEtran}
\bibliography{refs}

 




\vfill

\end{document}

%% file: sections/Introduction.tex
\section{Introduction}
\IEEEPARstart{M}{echatronic} motion systems are pivotal components in industrial processes and manufacturing applications, including wafer scanners, industrial printers, pick-and-place machines, and wire bonders~\cite{butler2011position,ohnishi1996motion,oomen2018advanced,blanken2019multivariable,weekers2024data,kostic2004modeling}. These systems are becoming increasingly complex due to flexible dynamics, position-dependent dynamics, and non-negligible interactions between motion axes in the bandwidth region~\cite{steinbuch2013design}. As the complexity of the dynamics and the number of inputs and outputs grow, parametric system identification or first-principle modelling become less effective for the construction of low-complexity but highly accurate dynamic models. This limitation underlines the need for a data-driven controller design that does not rely on parametric models~\cite{hou2013model}. In fact, the high-performance mechatronics industry has relied on \emph{frequency response function} (FRF) estimates characterizing the dynamics, rather than parametric models, as the basis for controller design~\cite{boerlage2005frequency}.

Traditionally, FRF-based controller design methods have been grounded in Cauchy's principle argument to ensure stability~\cite{broens2024frequency} and using mixed-sensitivity shaping to optimize performance~\cite{van2002multivariable}. Non-convex controller tuning techniques, such as \emph{particle swarm optimization} (PSO)~\cite{kennedy1995particle, xin2013self}, have proven effective for FRF-based controller tuning when only a limited number of design parameters are involved. However, as the complexity of the machine dynamics increases, so does controller complexity, and such methods do not scale efficiently for \emph{multiple-input multiple-output} (MIMO) systems. This highlights the need for \emph{convex} alternatives that enable computationally feasible controller synthesis based on FRF data.

In their recent work \cite{schuchert2024data}, the authors successfully combine MIMO stability and performance in a convex approach based on FRF data. The proposed design approach is iterative in nature, with stability guaranteed through 
an initial, stabilizing controller, while performance is optimized at each iteration. This approach assumes the availability of \emph{exact FRF data} over the \emph{entire frequency spectrum}. Several extensions have been proposed that account for FRF estimation errors by modeling those using an ellipsoidal uncertainty in a data-driven manner: in \cite{gupta2023data} for \emph{single-input single-output} (SISO) systems, and in \cite{gupta2024non, gupta2023hybrid} for MIMO systems. Another important consideration is the sampling of the spectrum. Since FRFs obtained from non-parametric identification are only available at finitely many frequency points, the intersample behavior of the system is unknown. This uncertainty aspect must be explicitly addressed to ensure that stability and performance guarantees provided in the design phase persist during deployment on the real system. In \cite{schuchert2024robotarm}, the spectrum sampling problem is tackled by introducing additional positive definite constraints on a multiplier, leveraging a \emph{Kalman-Yakubovich-Popov} (KYP) argument to guarantee robust stability and performance. However, estimation errors are not considered in that work. In contrast, \cite{schuchert2024robust} discusses both sampling and estimation errors for SISO systems. The authors take an uncertainty ball around each FRF data point, providing finite constraints to characterize the convex hull, thereby achieving robust guarantees.

Despite these advances, existing methods either treat estimation error or spectrum sampling as isolated problems, or are restricted to the SISO settings. For MIMO systems, the interplay between uncertainty and inter-sample behavior becomes significantly more intricate, since stability and performance must hold simultaneously across multiple interacting channels. This calls for a framework that can unify both aspects, namely estimation error and spectral sampling, within a single, convex, and computationally tractable design procedure.

Building on the existing promising approaches, the present paper introduces a convex controller design approach for MIMO systems that explicitly accounts for both estimation error and spectrum sampling. To this end, we extend prior work \cite{bloemers2022frequency, bloemers2021frequency, spin2026iteratively} on iterative, convex MIMO controller design based on FRF data with \emph{integral quadratic constraints} (IQCs). Introduced in the seminal work of \cite{megretski1997system}, IQCs offer a convex alternative to circle-counting methods, and originate from the multiplier approach \cite{zames1966input, zames1966input2}. The IQC approach is central to concepts such as absolute stability \cite{khalilnonlinear, popov1962absolute}. 
However, verifying IQCs requires considering infinitely many frequency points. As a result, enforcing the corresponding conditions only on sampled frequency points prevents rigorous guarantees between samples. Our proposed methodology addresses this challenge by combining the IQC framework with \emph{convex generator sets}. IQCs provide a unified and convex description of input–output properties of uncertain operators, while generator sets offer a tractable representation of FRF uncertainty regions at discrete frequencies and allow interpolation across spectral intervals. Together, these ingredients enable the formulation of robust stability and performance guarantees on the basis of finitely many FRF samples, while preserving convexity in the synthesis problem.

The main contributions of this paper can be summarized as follows: 
\begin{itemize} 
	\item A novel relaxation approach of infinite-dimensional IQC conditions to finitely many \emph{linear matrix inequalities} (LMIs), achieved through convex generator sets that interpolate FRF uncertainty between sampled frequencies; 
	\item A convex, iterative structured MIMO controller design approach that explicitly incorporates both estimation errors and spectrum sampling, thereby extending existing FRF-based approaches to the MIMO setting; 
	\item Experimental validation of the proposed MIMO motion-control design on an industrial grade wire bonder, demonstrating that the proposed method enables robust and high-performance control in a real-life mechatronic application. 
\end{itemize}

The paper is organized as follows. Section~\ref{sec:probdef} presents the problem formulation. Section~\ref{sec:idealSol} reviews the necessary preliminaries on integral quadratic constraints (IQCs) and develops an ideal theoretical solution. Section~\ref{sec:robSol} contains the main result, namely, a robust sampled-data solution with stability and performance guarantees. The generator construction required by the main result is developed in Section~\ref{sec:genConstr} for both uncertain and stochastic FRF estimates. Section~\ref{sec:example} demonstrates the proposed approach on an industrial-grade wire bonder, and Section~\ref{sec:conclusion} concludes the paper.

\noindent
{\bf Notation:} The sets $\mb{C}, \mathbb{R}, \mb{Z}$, and $\mathbb{N}$ denote the set of complex numbers, real numbers, integers, and non-negative integers (including zero), respectively. Matrices with dimension $n \times m$ over the field $\mb{K}$ are denoted by $\mb{K}^{n \times m}$. The notation $\mb{K}^{\sbullet \times \sbullet}$ is used when the dimensions are not relevant. Real symmetric and complex Hermitian matrices are denoted by $\mb{S}^\sbullet$ and $\mb{H}^\sbullet$, respectively. The notation $A \succ 0 (\succeq 0)$ and $A \prec 0 (\preceq 0)$ denotes that $A \in \mb{H}^\sbullet$ is positive (semi)-definite and negative (semi)-definite, respectively. If $0 \prec A \in \mb{R}^{\sbullet \times \sbullet}$ we write $A \in \mb{S}_{\succ 0}^\sbullet$. Definiteness of Hermitian matrices is similarly defined. We denote by $\mb{D}, \partial \mb{D}, \bar{\mb{D}},$ and $\mb{E}$ the open unit disc, the unit circle, the closed unit disc, and the complement of the closed unit disc in $\mb{C}$, respectively. The set of rational matrix functions over the field $\mb{R}$ is denoted by $\mb{R}^{\sbullet \times \sbullet}(\xi)$, where $\xi$ is the indeterminate. Let $\mc{G}$ be a discrete-time LTI system for which we use $G(z)$ with $z\in\mb{C}$ to denote its transfer function. We write $G \in \mc{R}_\mathrm{p}^{\sbullet \times \sbullet}$ if $G(\xi) \in \mb{R}^{\sbullet \times \sbullet}(\xi)$ and $G$ is a (strictly) proper function. If $G \in \mc{R}_\mathrm{p}^{\sbullet \times \sbullet}$ is analytic on the open set $\mb{E}$ and bounded on $\partial \mb{D}$, i.e., $ \sup_{\omega\in[-\pi,\pi]}\bar{\sigma} (G(e^{j\omega}))$ is bounded with $j$ the imaginary unit and $\bar{\sigma}$ is the maximum singular value, we write $G \in \mc{RH}_{\infty}^{\sbullet \times \sbullet}(\mb{E})$. For two systems $G_1$ and $G_2$, we use $G_1 \star G_2$ to denote a Redheffer product between $G_1$ and $G_2$, see \cite{zhou1996robust}. For a transfer function $G(z)$ evaluated on $\Omega := [-\pi, \pi]$, we have $G^*(e^{j \omega}) = G^{\top}(e^{-j \omega})$ for all $\omega \in \Omega$, where $^\ast$ denotes the conjugate transpose and $^\top$ stands for the transpose. Discrete-time (rational) transfer matrices with bounded spectra are denoted as $G \in \mc{(R)L}_{\infty}^{\sbullet \times \sbullet} (\partial \mb{D})$. 
If it holds that $G(e^{j \omega}) \prec 0$ for all $\omega \in \Omega$, we write $G \prec 0$ on $\partial \mb{D}$. 
We denote by $\mathcal{L}_2$ the space of square integrable real signals, while $\mathcal{L}_{2\rm{e}}$ stands for its extensions in terms of signals that are square integrable over any finite interval.
For brevity, we write quadratic products of the form $x^{\!\top}\! A x$ as $(\sbullet)^{\!\top}\!\! A x$. For a matrix $A\in\mathbb{C}^{n\times m}$, $\mr{vec}(A)$ denotes a concatenation of the columns of $A$ in a row vector while $\otimes$ denotes the Kronecker product. 
Furthermore, $\mathbb{I}_{\tau_1}^{\tau_2}=\{ \tau \in\mathbb{Z} \mid \tau_1 \leq \tau \leq \tau_2 \}$ is an index set.

%% file: sections/ProblemFormulation.tex
\section{Problem definition} \label{sec:probdef}

The generalized plant concept is a systematic approach for formulating analysis and controller synthesis problems for LTI systems even under uncertainty \cite{zhou1996robust}. Through this concept, a wide range of system and control configurations can be described, where performance specifications are encoded with weighting filters in an intuitive manner. In this work, we also consider the generalized plant concept to formulate control configurations and desired performance specifications. An example of a generalized plant connected with a controller $K$ is depicted on the left in Figure~\ref{fig:3rob}.

Consider a nominal minimal LTI \emph{discrete-time} (DT) generalized-plant $P$, given by the state-space representation 
\begin{equation}
	P := \left\{
	\begin{aligned}
		x_{t+1} &= Ax_t + B_\mr{w} w_t + B_\mr{u} u_t,\\
		z_t   &= C_\mr{z} x_t + D_\mr{zw} w_t + D_\mr{z u} u_t,\\
		y_t   &= C_\mr{y} x_t + D_\mr{yw} w_t + D_\mr{yu} u_t,
	\end{aligned}
	\right. \label{eq:ssrepr}
\end{equation}
where $t\in\mb{Z}$ is discrete time, $x_t \in \mb{R}^{n_\mr{x}}$ is the state, $w_t \in \mb{R}^{n_\mr{w}}$ is the generalized disturbance, $z_t \in \mb{R}^{n_\mr{z}}$ is the generalized performance output, while $y_t \in \mb{R}^{n_\mr{y}}$ and   $u_t \in \mb{R}^{n_\mr{u}}$ are the measured output and actuated input available for the controller, all at discrete time $t \in \mb{Z}$, respectively.  The FRF of~\eqref{eq:ssrepr} is defined as
\begin{equation}
	P(e^{j \omega}) =
	\Matrix{c:c}{
		P_\mr{zw}(e^{j \omega}) & P_\mr{zu}(e^{j \omega}) \\ \hdashline
		P_\mr{yw}(e^{j \omega}) & P_\mr{yu}(e^{j \omega})
	},
	\quad \omega \in \Omega,
\end{equation}
where $P \in \mc{R}_\mr{p}^{(n_\mr{z}+n_\mr{y}) \times (n_\mr{w}+n_\mr{u})}$ is a proper real-rational transfer function associated with~\eqref{eq:ssrepr}:
\begin{equation}
	P(z) = zI \star
	\Matrix{c:cc}{
		A & B_\mr{w} & B_\mr{u} \\ \hdashline
		C_\mr{z} & D_\mr{zw} & D_\mr{zu} \\
		C_\mr{y} & D_\mr{yw} & D_\mr{yu}
	}.
\end{equation}
It is assumed that the generalized plant $P$ explicitly encodes the chosen control interconnection (e.g., feedback and feedforward structure) together with the selected exogenous and performance input--output channels. In addition, $P$ incorporates the weighting filters that define the desired performance specifications (e.g., mixed-sensitivity weights), such that $P$ fully specifies both the closed-loop architecture and the associated performance objectives.

\begin{figure}[t]
	\vspace*{-5pt}
	\centering
	\includegraphics[width=0.98\linewidth]{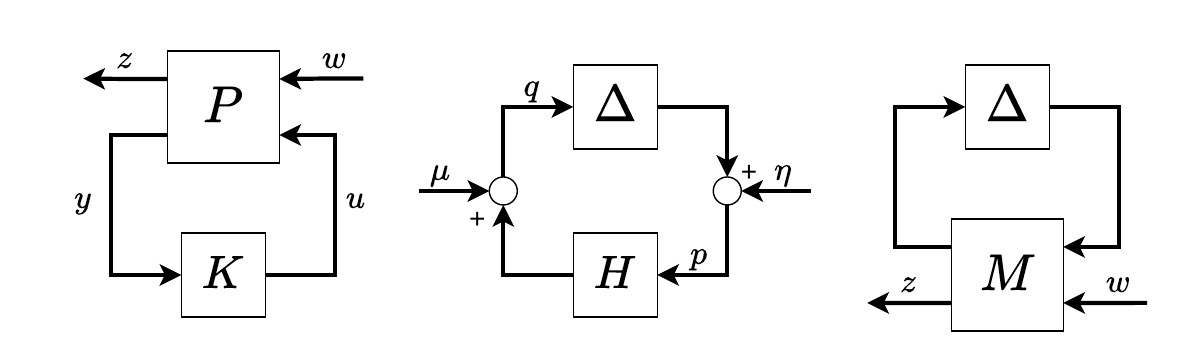}
	\vspace*{-5pt}
	\caption{Feedback interconnections used throughout the paper: Generalized plant $P$ and controller $K$ interconnection on the left, interconnection between an LTI system $H$ and an uncertainty $\Delta$ for analyzing robust stability in the middle, and interconnection between an LTI system $M$ and an uncertainty $\Delta$ for analyzing robust performance on the right.}
	\label{fig:3rob}
\end{figure}
The objective is that, based solely on FRF data\footnote{Let $G$ denote the open-loop plant to be controlled. Given FRF data $G(e^{j\omega})$, the generalized plant $P(e^{j\omega})$ is obtained by interconnecting $G(e^{j\omega})$ with the prescribed control configuration and weighting filters. Hence, $P(e^{j\omega})$ can be evaluated pointwise in frequency without requiring a parametric model of $G$.} of $P(e^{j \omega})$, design a fixed-structure MIMO feedback controller $K\in \mc{R}_\mr{p}^{ n_\mr{u} \times n_\mr{y}}$ such that the closed-loop transfer function from $w$ to $z$,
\begin{equation}
	T_\mr{zw} = P \star K := P_\mr{zw} + P_\mr{zu} K \left( I - P_\mr{yu} K \right)^{-1} P_\mr{yw},
\end{equation}
is proper, asymptotically stable, and meets the prescribed performance specifications.  
Correspondingly, the design objective can be formulated as
\begin{subequations} \label{eq:des}
	\begin{align} \label{eq:des_a}
		\underset{K}{\mr{minimize}} \quad & \|T_\mr{zw}\|_{\mc{H}_p}, \quad p \in \{2, \infty\}, \\
		\mr{subject \ to} \quad & T_\mr{zw} \in \mc{RH}_\infty^{n_\mr{z} \times n_\mr{w}},
	\end{align}
\end{subequations}
where the performance metrics in \eqref{eq:des_a} are given by
\begin{subequations}
	\begin{align}
		\|T_\mr{zw}\|_{\mc{H}_\infty} &= \sup_{\omega \in \Omega} \bar\sigma \big(T_\mr{zw}(e^{j \omega})\big), \\
		\|T_\mr{zw}\|_{\mc{H}_2}^2 &= \frac{1}{2 \pi} \int_{-\pi}^{\pi} \mathrm{Tr} \left( T_\mr{zw}^*(e^{j \omega}) \, T_\mr{zw}(e^{j \omega}) \right) \mathrm{d}\omega.
	\end{align}
\end{subequations}
Convex formulations for this problem are available in~\cite{schuchert2024data,bloemers2021frequency}.  
However, these methods result in semi-infinite \emph{semi-definite programs} (SDPs), since the SPDs must hold for all $\omega \in \Omega$, and exact knowledge of $P(e^{j \omega})$ is required.

In this work, we instead assume that only estimates $\widehat{P}(e^{j \omega})$ of $P(e^{j \omega})$ with uncertainty are available at a finite set of frequency points $\Omega_N \subset \Omega$, where $\Omega_N = \{ \omega_k \}_{k=0}^{N-1} \subset [0, \pi]$ and
\begin{equation}\label{eq:OmegaN}
	0 = \omega_0 < \omega_1 < \dots < \omega_{N-1} = \pi.
\end{equation}
At each $\omega_k \in \Omega_N$, we consider the possible realizations of $\widehat{P}(e^{j \omega})$ as $\widehat{\mathcal{P}}(e^{j \omega_k})\subset \mathbb{C}^{(n_\mathrm{z} + n_\mathrm{y}) \times (n_\mathrm{w} + n_\mathrm{u})}$ being a convex  complex set  such that
\begin{equation}\label{eq:estimatesFinPoints}
	P(e^{j \omega_k}) \in \widehat{\mathcal{P}}(e^{j \omega_k}), \quad \forall \, \omega_k \in \Omega_N, \ k\in\mathbb{I}_{0}^{N-1}.
\end{equation}
Note that \eqref{eq:estimatesFinPoints} is a rather mild and realistic assumption in practice. Indeed, in case of frequency domain indentification techniques aimed at estimation of FRFs, an ellipsoidal uncertainty region characterized by a confidence level and the covariance of the resulting estimates can be directly obtained by standard techniques \cite{pintelon2012system}, see also Section \ref{sec:stoc:gen}.

As pointed out in the introduction, the inter-sample behavior—i.e., the frequency response function in between the available frequency points—is largely neglected in the existing literature as it is a challenging problem to provide performance and stability guarantees in terms of \eqref{eq:des} for the true underlying system when only information in terms of $\{\widehat{\mathcal{P}}(e^{j \omega_k})\}_{k=0}^{N-1}$ at frequency points $\omega_k \in \Omega_N$ is available about the plant. The assumptions on the inter-sample behavior will be made explicit in Section~\ref{sec:robSol}, where they are constructed directly from the available finite and uncertain FRF data. This challenge motivates the following research question:

	\emph{How can robust, structured MIMO controllers be designed from an uncertain FRF estimate  $\{\widehat{\mathcal{P}}(e^{j \omega_k})\}_{k=0}^{N-1}$ of a real system at finite number of frequency points $\omega_k \in \Omega_N$, while ensuring performance and stability guarantees, also for the inter-frequency-sample behavior?}
	
In the sequel, we develop a convex robust synthesis approach to address this research question.

%% file: sections/SolutionApproach.tex
\section{Convex Synthesis Under Full and Exact Spectral Information} \label{sec:idealSol}
%
In order to develop our solution, first we introduce a convex synthesis method under the ideal assumption that the exact frequency response function $P(e^{j\omega})$ of the plant is known over the entire spectrum $\Omega$. This setting allows us to focus on the structural aspects of the synthesis procedure without considering practical limitations such as spectrum sampling or uncertainty in the FRF estimates yet. By temporarily setting aside computational tractability, we show how the synthesis problem can be formulated as a convex optimization program. This ideal formulation reveals the role of coprime factor parameterizations and enables the use of IQCs to characterize robust stability and performance. \vspace{-4mm}

\subsection{Preliminaries on Integral Quadratic Constraints}
Consider the interconnection between a discrete-time LTI system $H\in\mc{RH}_\infty^{n_\mr{q} \times n_\mr{p}}$ and a causal bounded operator $\Delta: \mathbb{R}^{n_\mr{q}} \rightarrow \mathbb{R}^{n_\mr{p}}$, depicted  in the middle of Figure~\ref{fig:3rob}. In general, $\Delta$ may be infinite-dimensional, nonlinear, and parametric/dynamic uncertainty. In the robust control framework~\cite{megretski1997system}, uncertainties are typically modeled via a class $\boldsymbol{\Delta}$ whose elements share prescribed properties, such as bounded gain, passivity, or structural constraints (e.g., block-diagonal form). The following definitions~\cite{veenman2016robust} are standard for such interconnections.
\begin{definition}[Well-posedness]
	The feedback interconnection $\Delta \star H$ is well-posed, if for each $\mr{col}(\mu, \eta) \in \mc{L}_{2\mr{e}}^\sbullet$ there exists a unique response $\mr{col}(q, p) \in \mc{L}_{2 \mr{e}}^\sbullet$ satisfying $q = H p + \mu$ and $p = \Delta(q)+\eta$ such that the map $\mr{col}(\mu, \eta) \rightarrow \mr{col}(q, p)$ is causal.
\end{definition}
\begin{definition}[Stability]
	The feedback interconnection $\Delta \star H$ is $\mc{L}_2$-stable if it is well-posed and if the induced $\mc{L}_2$-gain of the map  $\mr{col}(\mu, \eta) \rightarrow \mr{col}(q, p)$ is bounded.
\end{definition}
\begin{definition}[Robust stability]
	The feedback interconnection $\Delta \star H$ is robustly $\mc{L}_2$-stable if it is well-posed and $\mc{L}_2$-stable for all $\Delta \in \boldsymbol{\Delta}$.
\end{definition}
\begin{remark}
	In the special case where $\Delta$ is an finite-dimensional LTI system, the $\mc{L}_2$-stability notion coincides with asymptotic stability.
\end{remark}
Stability of the feedback interconnection $\Delta \star H$ is characterized directly in terms of input--output properties of the uncertainty class $\boldsymbol{\Delta}$. Rather than relying on a specific representation of $\Delta$, we employ the framework of Integral Quadratic Constraints (IQCs)~\cite{hu2017robustness} to capture these properties.
\begin{definition}[IQC]
	The DTFT  $\mr{col}(q, p) \in \mc{L}_{\infty}^{n_\mr{q} + n_\mr{p}}(\partial \mb{D})$ of the corresponding time signals satisfy the frequency-domain IQC, with respect to a so-called IQC-multiplier\footnote{In general, the IQC-multiplier $\Pi \in \mc{L}_\infty^{(n_\mr{q}+n_p) \times (n_\mr{q}+n_p)}(\partial \mb{D})$ can be any measurable essentially bounded self-adjoint Hermitian-valued function. } $\Pi(e^{j \omega}) \in \mb{H}^{n_\mr{q}+n_p}$, $\forall\omega \in \Omega$, if $\mc{I}(\Pi, q, p) \geq 0$ where
	\begin{equation*}
		\mc{I}(\Pi, q, p) := \int_{-\pi}^\pi \Matrix{c}{q(e^{j \omega}) \\ p(e^{j \omega})}^* \Pi(e^{j \omega}) \Matrix{c}{q(e^{j \omega}) \\ p(e^{j \omega})} \mr{d} \omega.
	\end{equation*}
\end{definition}
For reasons that will become apparent in the subsequent analysis, we also consider the case in which $0 \notin \boldsymbol{\Delta}$. This necessitates the use of a specialized class of IQC multipliers capable of handling such uncertainty sets, notably multipliers subject to inertia constraints; see, e.g.,~\cite{veenman2013stability}. The inertially constraint multiplier class is denoted as:
\begin{multline}\label{eq:inertiaMult}
	\boldsymbol{\Pi}_\mr{in} := \bigg\{ \Pi = \Pi^* \in \mc{L}_\infty^{(n_\mr{q}+n_\mr{p}) \times (n_\mr{q}+n_\mr{p})}(\partial \mb{D}) \ \mr{s.t.} \\ \Pi = \Matrix{cc}{\Pi_\mr{qq} & \Pi_\mr{q p} \\ \sbullet & \Pi_\mr{p p}}, \ \Pi_\mr{qq} \in \mb{H}^{n_\mr{q}}_{\succ 0}, \ \Pi_\mr{pp} \in \mb{H}^{n_\mr{p}}_{\preceq 0} \ \mr{on} \ \partial \mb{D} \bigg\}.
\end{multline}
For $\Pi \in \boldsymbol{\Pi}_\mr{in}$, we may also write $\Pi(e^{j \omega}) \in \mb{H}_\mr{in}^{(n_\mr{q}+n_\mr{p}) \times (n_\mr{q}+n_\mr{p})}$ for all $\omega \in \Omega$, where $\mb{H}_\mr{in}$ is implicitly defined through~\eqref{eq:inertiaMult}. This class of multipliers, also known as PN multipliers~\cite{seiler2010dissipation}, allows the standard IQC-based robust stability result~\cite{veenman2013stability} to be formulated in a relaxed form without resorting to the homotopy argument~\cite{megretski2001new}.

In general, an uncertainty class $\boldsymbol{\Delta}$ is characterized within the IQC framework by an associated multiplier class $\boldsymbol{\Pi}_\Delta$, whose elements encode the admissible input--output properties of all $\Delta\in\boldsymbol{\Delta}$; see~\cite{veenman2016robust} for an extensive catalog covering, among others, structured and unstructured uncertainties, passive nonlinearities, and bounded real systems. When $0 \notin \boldsymbol{\Delta}$, robust stability can be certified by restricting attention to multipliers in the intersection $\boldsymbol{\Pi}_\Delta \cap \boldsymbol{\Pi}_\mr{in}$. This leads to the following relaxed IQC-based robust stability result.
\begin{theorem}[Robust stability via IQCs]
	Let $H\! \in\! \mc{RH}_\infty^{\sbullet\! \times\! \sbullet}\!(\mb{E})$ and suppose the interconnection $\Delta \star H$ is well-posed for all $\Delta\in \boldsymbol{\Delta}$. If the IQC $\mc{I}(\Pi^\mr{s}, q, \Delta(q)) \geq 0$ is satisfied for all $\Delta \in \boldsymbol{\Delta}$ and all $\Pi^\mr{s} \in \boldsymbol{\Pi}_\Delta \cap \boldsymbol{\Pi}_\mr{in}$ and the \emph{frequency domain inequality} (FDI) 
	\begin{equation}
		\Matrix{c}{H \\ I}^* \Pi^\mr{s} \Matrix{c}{H \\ I} \prec 0 \quad \mr{on} \ \partial \mb{D}.
	\end{equation}
	holds for some $\Pi^\mr{s} \in \boldsymbol{\Pi}_\Delta \cap \boldsymbol{\Pi}_\mr{in}$, then the interconnection $\Delta \star H$ is robustly stable.
\end{theorem}
To extend this result for robust performance, consider the feedback interconnection shown in Figure~\ref{fig:3rob}. Robust performance is analyzed with a performance multiplier class $\boldsymbol{\Pi}_\mr{p}$, see~\cite{veenman2016robust} for an extensive overview of performance metrics expressed using multiplier classes. In this paper, we focus on $\mc{H}_\infty$ and $\mc{H}_2$ performance which are given by the performance multipliers classes $\boldsymbol{\Pi}_{\mc{H}_\infty} := \{\mr{diag}(\gamma^{-1} I_{n_\mr{z}}, \ - \gamma I_{n_\mr{w}}) \ | \ \gamma \geq 0\}$ and $\boldsymbol{\Pi}_{\mc{H}_2} := \{\mr{diag}(I_{n_\mr{z}}, -Q(e^{j \omega})) \ | \ Q(e^{j \omega}) \in \mb{H}^{n_\mr{z}}_{\succ 0}$ for all $\omega \in \Omega \}$, respectively. The associated cost functions (to be minimized) are $J(\boldsymbol{\Pi}_{\mc{H}_\infty}) = \gamma$ and $J(\boldsymbol{\Pi}_{\mc{H}_2}) = \sqrt{\frac{1}{2 \pi} \int_{-\pi}^{\pi} \mr{Tr} \left(Q(e^{j \omega}) \right) \mr{d}\omega}$. To simultaneously certify robust stability and robust performance, i.e., a given performance bound for all possible realizations of the uncertainty, we combine the corresponding stability and performance multipliers into the augmented multiplier
\begin{equation}\label{eq:daug}
	\begin{aligned}
		\Pi
		=
		\mr{daug}(\Pi^\mr{s}, \Pi^\mr{p})
		:=
		\Matrix[2]{cc:cc}
			{\Pi_\mr{qq}^\mr{s} & 0 & \Pi_\mr{qp}^\mr{s} & 0 \\
				0 & \Pi_\mr{zz}^\mr{p} & 0 & \Pi_\mr{zw}^\mr{p} \\
				\hdashline
				(\Pi_\mr{qp}^\mr{s})^* & 0 & \Pi_\mr{pp}^\mr{s} & 0 \\
				0 & (\Pi_\mr{zw}^\mr{p})^* & 0 & \Pi_\mr{ww}^\mr{p}},
	\end{aligned}
\end{equation}
where $\Pi^\mr{s}\in\boldsymbol{\Pi}_\Delta$ and $\Pi^\mr{p}\in\boldsymbol{\Pi}_\mr{p}$. The following analysis result~\cite{veenman2016robust} provides a frequency-domain method to characterize both.
\begin{theorem}[Robust performance via IQCs]
	Let $M \in \mc{RH}_\infty^{\sbullet \times \sbullet}(\mb{E})$ and suppose the interconnection between $M$ and $\Delta$ is well-posed for all $\Delta\in \boldsymbol{\Delta}$. If the IQC  $\mc{I}(\Pi^\mr{s}, q, \Delta(q)) \geq 0$ is satisfied for all $\Delta \in \boldsymbol{\Delta}$ and all $\Pi^\mr{s} \in \boldsymbol{\Pi}_\Delta \cap \boldsymbol{\Pi}_\mr{in}$ and the FDI 
	\begin{equation}\label{eq:robstabperfFDI}
		\! \Matrix{cc}{M \\ I}^* \Pi \Matrix{c}{M \\ I} \prec 0 \ \mr{on} \ \partial \mb{D}.
	\end{equation}
	holds for some $\Pi^\mr{s} \in \boldsymbol{\Pi}_\Delta \cap \boldsymbol{\Pi}_\mr{in}$ and some $\Pi^\mr{p} \in \boldsymbol{\Pi}_\mr{p} \cap \boldsymbol{\Pi}_\mr{in}$, then the interconnection is robustly stable and satisfies robust performance.
\end{theorem}
\begin{remark}
	Most results hold also for continuous-time systems by small changes. Moreover, we do not assume that the multipliers are rational, which reduces conservatism compared to state-space approaches where this assumption is needed.
\end{remark}
The IQC framework thus provides a rigorous and general foundation for analyzing the robust stability and performance of uncertain systems. Next, we show how it can be used for direct frequency-domain synthesis based on $P(e^{j\omega})$ without an analytic model of the plant.  
\vspace{-4mm}
\subsection{Frequency-domain IQC synthesis method}

Here, we provide a further generalization of the approach in~\cite{bloemers2021frequency} for the discrete-time case and for general-plant interconnections. 
We begin by parametrizing the to-be-designed controller $K$ using a \emph{right coprime factorization} 
(rcf) in terms of the pair $(V, U)$. The pair $(V, U)$ is a rcf of $K$ over $\mc{RH}_\infty(\mb{E})$ if both $V, U \in \mc{RH}_\infty^{\sbullet \times \sbullet}(\mb{E})$, $V$ is invertible, $K = UV^{-1}$, and there exist $X, Y \in \mc{RH}_\infty^{\sbullet \times \sbullet}(\mb{E})$ such that $XV + YU = I$. For further details see~\cite{vidyasagar2022control}. The FRFs of the coprime pair $(V, U)$ can be parametrized as
\begin{subequations}
	\begin{align}
		V(e^{j \omega}) &= C_V \phi(e^{j\omega}) + D_V, \\
		U(e^{j \omega}) &= C_U \phi(e^{j\omega}) + D_U,
	\end{align}
\end{subequations}
where $\phi(e^{j \omega}) := \left( e^{j \omega} I - A_\phi \right)^{-1} B_\phi \in \mc{RH}_{\infty}^{n_\phi \times n_\mr{y}}$, $V \in \mc{RH}_\infty^{n_\mr{y} \times n_\mr{y}}$, and $U \in \mc{RH}_\infty^{n_\mr{u} \times n_\mr{y}}$. The state-space matrices $(A_\phi, B_\phi)$ can be chosen based on the realization of rational \emph{Generalized Orthogonal Basis Functions} (GOBF)~\cite{van2000modelling}, representing stable, orthogonal bases. The pair $(A_\phi, B_\phi)$ is fixed and therefore, the FRFs of the pair $(V, U)$ are affine functions of the optimization variables $\theta := \mathrm{vec}(C_V, C_U, D_V, D_U)$ with $\theta\in\Theta=\mathbb{R}^{d}$, $d = (n_\phi + n_\mr{y})(n_\mr{u}+n_\mr{y})$, which is essential for formulating convex optimization programs. For a given $\theta$, the controller $K$ can be realized as
\begin{equation} \label{eq:controllerRealization}
	K(z) = zI \star \Matrix{c:c}{A_\phi - B_\phi D_V^{-1} C_V & B_\phi D_V^{-1} \\
		\hdashline
		\noalign{\vskip 0.5ex} 
		C_U - D_U D_V^{-1} C_V & D_U D_V^{-1}}.
\end{equation}
Additional sparsity in the elements of $K$, e.g., a distributed controller parametrization, can be incorporated by adding constraints on the optimization variables $\theta$, see~\cite{schuchert2024data}. 

To illustrate the structure of the closed-loop system under this parametrization, we consider a lcf  
of $P_\mr{yu} = D^{-1} N$ in combination with the 
rcf parametrization of the controller. This leads to the closed-loop description
\begin{equation} \label{eq:Mdef}
	T_\mr{zw} := P \star K = P_\mr{zw} + P_\mr{zu} U \Lambda^{-1} D P_\mr{yw},
\end{equation}
where $\Lambda := DV - NU$ and $P$ is the generalized plant. 
Using the coprime factors, the closed-loop system can be equivalently expressed as the \emph{linear fractional transformation} (LFT)
\begin{equation}\label{eq:CLasIQC}
	T_\mr{zw} = \Delta_I \star \underbrace{\Matrix{cc}{I-\Lambda\ &\ D P_\mr{yw} \\ P_\mr{zu} U\ &\ P_\mr{zw}}}_{=:M},
\end{equation}
where $\Delta_I := I$.
We explicitly use the notation $\Delta_I$ to emphasize that the identity can be interpreted as a special virtual uncertainty from a set $\boldsymbol{\Delta}_I$ containing a single element. To certify the robust stability of the interconnection~\eqref{eq:CLasIQC} using IQC analysis, an associated multiplier class for the virtual uncertainty set is required. Following~\cite{bloemers2021frequency}, the associated multiplier class is given by
\begin{equation}
	\boldsymbol{\Pi}_I :=
	\left\{
	\Pi \in \boldsymbol{\Pi}_\mr{in}
	\ \bigg| \
	\Matrix{c}{I\\I}^{\!*}
	\Pi
	\Matrix{c}{I\\I}
	\succeq 0
	\quad
	\text{on } \partial\mathbb{D}
	\right\}.
\end{equation}
%
%
By construction, we have $M \in \mc{RH}_{\infty}^{\sbullet \times \sbullet}(\mb{E})$ which depends affinely on the controller parameters $\theta$. Given a known right coprime factorization $(V, U)$ of a stabilizing controller $K$, e.g., a fixed $\theta$ is available, the following optimization program can be used to compute an IQC multiplier $\Pi$:
\begin{equation}\label{eq:MultIQCite}
\!\!\!\!\!\!	\begin{aligned}
		&\underset{\Pi^\mr{s} \in \boldsymbol{\Pi}_I, \ \Pi^{\mr{p}} \in \boldsymbol{\Pi}_\mr{{p}}}{\mathrm{minimize}} & & \ \ \ \ \ J(\boldsymbol{\Pi}_\mr{{p}}), \\
		&  \! \ \ \ \ \mr{subject \ to} & & \ \! \Matrix{c}{M \\ I}^* \Pi \Matrix{c}{M \\ I} \prec 0 \ \mr{on} \ \partial \mb{D},
	\end{aligned}
\end{equation}
where the IQC multiplier is given by $\Pi = \mr{daug}(\Pi^\mr{s}, \Pi^\mr{p})$. Let this augmented multiplier be partitioned as
\begin{equation}
	\begin{aligned}
		\Pi = \Matrix{c:c}{\Pi_{11} & \Pi_{12} \\ \hdashline \Pi_{12}^* & \Pi_{22}}. 
	\end{aligned}
\end{equation}
Using Schur's complement, the FDI in the constraint of optimization program~\eqref{eq:MultIQCite} can be equivalently expressed as
\begin{equation}\label{eq:SchurFDIConstraints}
	\Matrix{cc}{\Pi_{22}+M^* \Pi_{12} + \Pi^*_{12} M & M^* \\ M & -\Pi_{11}^{-1}} \prec 0 \ \mr{on} \ \partial \mb{D}.
\end{equation} 
Given a known stability multiplier $\Pi^\mr{s}$, the following optimization program can be solved over the controller parameters $\theta$:
\begin{equation}\label{eq:ContIQCite}
	\begin{aligned}
		\underset{\theta\in\Theta,\ \Pi^{\mr p}\in\boldsymbol{\Pi}_{\mr p}}{\mr{minimize}}
		\quad & J\bigl(\boldsymbol{\Pi}_{\mr p}\bigr),\\
		\mr{subject\ to}
		\quad & \eqref{eq:SchurFDIConstraints}.
	\end{aligned}
\end{equation}
With these convex optimization programs, the controller synthesis problem can be solved by using Algorithm~\ref{alg:idealAlg}. This algorithm builds upon the frequency-domain framework introduced in~\cite{bloemers2021frequency}, extending it to the discrete-time case and accommodating general $\mc{H}_\infty$ and $\mc{H}_2$ performance specifications while allowing for a more flexible generalized plant formulation compared to the 4-block case treated in ~\cite{bloemers2021frequency}. 
\begin{algorithm}[t]
	\caption{Iterative IQC-Based Controller Synthesis}\label{alg:idealAlg}
	\begin{algorithmic}[1]
		\Require Initial stabilizing $K_0 = U_0 V_0^{-1}$, given by $\theta_0\in\Theta$
		\State {\bf solve} SDP \eqref{eq:MultIQCite} under fixed $\theta_{0}$, giving $\Pi_{0}$
		\State {\bf set} $J^* \gets J(\Pi_{\mathcal{H}_p})$ with $\Pi_{0},\theta_{0}$, $i\leftarrow 1$
		\Repeat
		\State {\bf solve} SDP \eqref{eq:ContIQCite} under fixed $\Pi_{i-1}$, giving $\theta_{i}$
		\State {\bf solve} SDP \eqref{eq:MultIQCite} under fixed $\theta_{i}$, giving $\Pi_{i}$
	        \State {\bf update}  $J^* \gets J(\Pi_{\mathcal{H}_p})$ with $\Pi_{i},\theta_{i}$
		\State {\bf update} controller $K \gets U_iV_i^{-1}$ based on $\theta_i$
		\Until{$J^*$ converged}
	\Ensure $K$, ensuring robust stability and performance specs 
	\end{algorithmic}
\end{algorithm}


To conclude, Algorithm~\ref{alg:idealAlg} provides a convex and structurally transparent synthesis procedure, yet it requires exact knowledge of the FRF $P(e^{j\omega})$ across the full spectrum $\Omega$, an assumption rarely satisfied in practice. Moreover, if such an FRF would be available, still the enforcement of constraints at infinitely many frequency points renders the method computationally infeasible. The next section addresses both issues by introducing a finite-data synthesis framework that explicitly incorporates uncertainty and spectrum sampling.

%% file: sections/MainResult.tex
\section{Convex Synthesis Under Sampled and Uncertain Spectral Information}\label{sec:robSol}

Based on the developments in Section~\ref{sec:idealSol}, this section presents the main theoretical contribution of the paper: a novel methodology that resolves the computational infeasibility inherent in frequency-domain IQC analysis, which ordinarily requires constraints to hold at infinitely many frequency points. Specifically, the result shows that robust stability and performance guarantees can be certified from finitely many FRF estimates, while simultaneously accounting for both the uncertainty at sampled frequencies and the inter-frequency sample behavior. Later in this section, we demonstrate how this result can be used to formulate a convex and computationally tractable controller synthesis procedure. 

\vspace{-2mm}
\subsection{Finite-dimensional IQC Certification}

The generalized plant is represented by the unknown frequency response $P(e^{j\omega}), \ \omega\in\Omega.$ At the finite set of measured frequencies $\Omega_N$ satisfying~\eqref{eq:OmegaN}, we denote by
\begin{equation}
	\widehat{P}_k=
	\Matrix{cc}{
		\widehat{P}_{\mathrm{zw},k} &
		\widehat{P}_{\mathrm{zu},k} \\
		\widehat{P}_{\mathrm{yw},k} &
		\widehat{P}_{\mathrm{yu},k}
	},
	\qquad
	k\in\mathbb{I}_0^{N-1},
\end{equation}
the mean (or nominal) estimate of the generalized plant, where we use the shorthand notation $\widehat P_k:=\widehat P(e^{j\omega_k})$. Associated with each estimate is a convex uncertainty region
\begin{equation}
	\widehat{\mathcal P}_k=
	\Matrix{cc}{
		\widehat{\mathcal P}_{\mathrm{zw},k} &
		\widehat{\mathcal P}_{\mathrm{zu},k}\\
		\widehat{\mathcal P}_{\mathrm{yw},k} &
		\widehat{\mathcal P}_{\mathrm{yu},k}
	},
\end{equation}
whose individual blocks are defined by
\begin{equation}
	\label{eq:Pblocks}
	\begin{aligned}
		\widehat{\mathcal P}_{\mathrm{zw},k}
		&=
		\operatorname{co}(\Gamma_{P,k}^{\mathrm{zw}}),
		&
		\widehat{\mathcal P}_{\mathrm{zu},k}
		&=
		\operatorname{co}(\Gamma_{P,k}^{\mathrm{zu}}),
		\\
		\widehat{\mathcal P}_{\mathrm{yw},k}
		&=
		\operatorname{co}(\Gamma_{P,k}^{\mathrm{yw}}),
		&
		\widehat{\mathcal P}_{\mathrm{yu},k}
		&=
		\operatorname{co}(\Gamma_{P,k}^{\mathrm{yu}}),
	\end{aligned}
\end{equation}
where $\Gamma_{P,k}^{\mathrm{zw}},\ldots,\Gamma_{P,k}^{\mathrm{yu}}$ denote sets of independent matrix-valued generators describing the uncertainty at frequency $\omega_k$. The construction of these generators for both uncertain and stochastic FRF estimates is discussed in the subsequent section, while Appendix~\ref{app:convex} reviews the required notions of convex sets and independent generators.
%


Similar to the linearly interpolated approach of~\cite{schuchert2024robust} for SISO uncertainties, we  require interpolation operators before presenting the main theorem. Toward definining these operators, we first formally define piecewise convexity of set-valued maps, which requires the partition $\Omega_N$ of $\Omega$, satisfying the assumptions in~\eqref{eq:OmegaN}. The frequency intervals associated to this partition are given by $\mc{I}_{k|k+1}:=[\omega_k, \omega_{k+1}]$.
\begin{definition}
	A set-valued map $\ms{S}: \Omega \mapsto \mb{K}^{n \times m}$ is said to be \emph{piecewise convex on $\Omega$ with respect to the partition} $\Omega_N$ if, for every interval $\mc{I}_{k|k+1}$, there exist polytopes $\mc{S}_{k},\mc{S}_{k+1} \subseteq \mb{K}^{n \times m}$ such that
	\begin{equation}
		\ms{S}(\omega)
		=
		(1-\lambda_\omega)\mc{S}_k
		+
		\lambda_\omega\mc{S}_{k+1},
	\end{equation}
	for $\omega\in\mc{I}_{k|k+1}$ and $k\in\mb{I}_0^{N-2}$.
\end{definition}

A piecewise convex extension of the estimated FRF of the generalized plant, given by~\eqref{eq:Pblocks}, is given by introducing a linear interpolation operator, defined as:
\begin{multline} 
	\ms{L}(\widehat{\mc{P}})(\omega) := \Bigl\{ (1-\lambda_\omega) \widetilde P_k + \lambda_\omega \widetilde P_{k+1} \;\Big|\; \widetilde P_k \in \widehat{\mathcal{P}}_k, \\ \widetilde P_{k+1} \in \widehat{\mathcal{P}}_{k+1}\Bigr\}, 
\end{multline}
for $\omega \in \mathcal{I}_{k|k+1}, \ k\in\mathbb{I}_0^{N-2},$ with the interpolation parameter given by
\begin{equation*}
	\lambda_\omega = \frac{\omega - \omega_k}{\omega_{k+1} - \omega_k}.
\end{equation*}
Here, the dependence of the interval $\mc{I}_{k|k+1}$ on the interpolation variable $\lambda_\omega$ is implicit.
Based on this construction, we introduce the set
\begin{equation}\label{eq:uncertainGenPlantSystemClass}
	\ms{P} := \Bigl\{ P \in \mc{R}_p^{\sbullet \times \sbullet} \;\big|\; P(e^{j \omega}) \in \ms{L}(\widehat{\mc{P}})(\omega), \ \forall \omega \in \Omega \Bigr\},
\end{equation}
which characterizes all generalized plants that must be stabilized by the controller $K$. The true generalized plant is assumed to satisfy $P \in \ms{P}$. For compactness, let $\ms{P}(\omega)$ denote the, generally nonconvex, family of sets $\{ P(e^{j \omega}) \mid P \in \ms{P}\}$ for $\omega\in\Omega$. The convex sets $\widehat{\mathcal{P}}(e^{j \omega_k})$, the family of sets $\ms{P}(\omega)$, and the interpolated sets $\ms{L}(\widehat{\mc{P}})(\omega)$ are depicted in Figure~\ref{fig:LinInterpPolytopicOverApprox}. 
\begin{figure}[t!]
	\centering
	\includegraphics[width=0.45\textwidth]{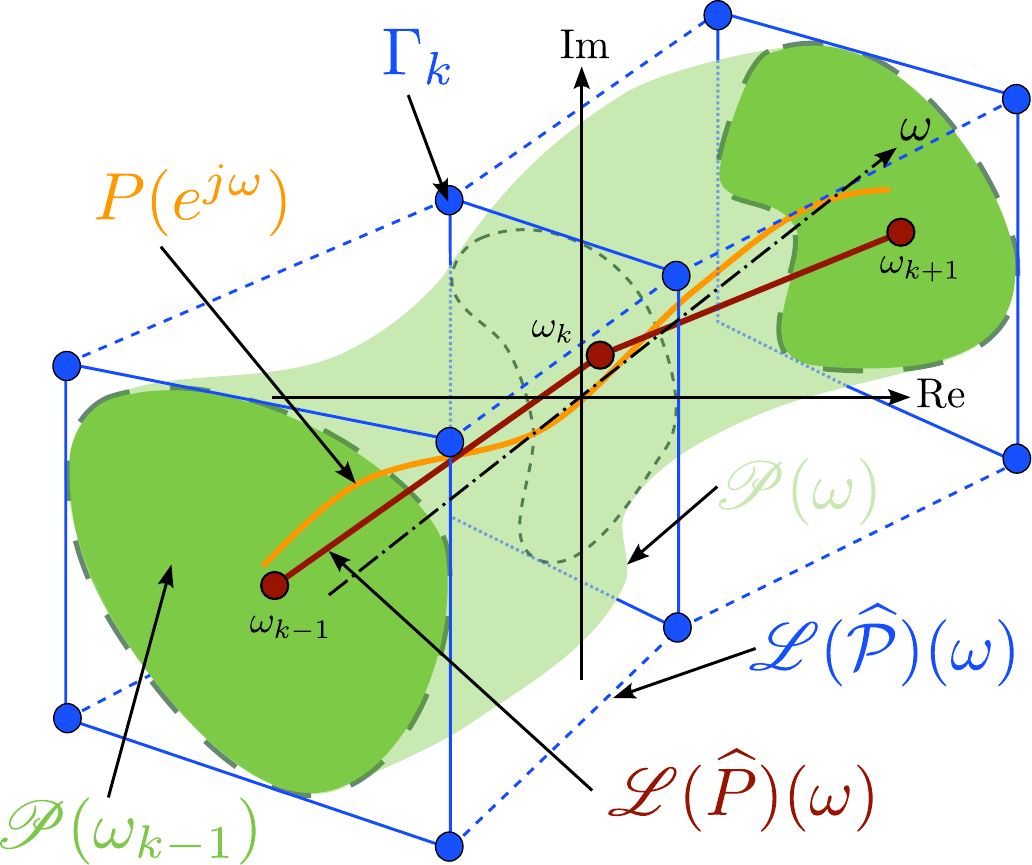}
	\caption{Conceptual figure showing the piecewise convex uncertain plant description. Real plant is given in orange, the mean estimates and their linear interpolation is given in dark red, the blue lines denote the piecewise convex region and the blue bullets the generators, the light green region denotes the uncertain FRF of the system and dark green denotes the uncertainty at measured frequency points $\omega_k$.} 
	\label{fig:LinInterpPolytopicOverApprox}
\end{figure}

Before presenting the main theorem, we introduce an interpolation operator that defines a convex region on each frequency interval $\mathcal{I}_{k|k+1} := [\omega_k, \omega_{k+1}]$, $k \in \mathbb{I}_0^{N-2}$, without enforcing continuity at the boundaries $\omega_k$. In particular, the left and right limits of $\omega$ approaching $\omega_k$ are not required to coincide. Such continuity at the boundary points is inherent to the previously defined linear interpolation operator $\mathscr{L}(\bullet)$; however, as will be shown in the subsequent section, constructing an analogous object to $M(e^{j\omega})$ in~\eqref{eq:CLasIQC} for the generalized plant $\mathscr{L}(\widehat{\mathcal{P}})(\omega)$ necessitates this more general notion. The inherent discontinuities at the boundary frequencies arise from sampling the coprime factors $(V,U)$ of the controller as will become evident in the subsequent section. 

Anticipating this, we introduce an interpolation concept that explicitly accommodates boundary discontinuities while preserving convexity on each interval. 
Specifically, each sampled frequency point $\omega_k \in \Omega_N$ belongs to both intervals $\mc{I}_{k-1|k}$ and $\mc{I}_{k|k+1}$. The definition of piecewise convexity, however, does not require the convex sets associated with these intervals to coincide at $\omega_k$. Consequently, different generator sets may be used depending on the interval under consideration. We therefore denote by $\Gamma_{\mc{S},k}^-$ the generator set at $\omega_k$ used in $\mathcal{I}_{k-1|k}$, and by $\Gamma_{\mc{S},k}^+$ the one used in $\mathcal{I}_{k|k+1}$. To accommodate such boundary discontinuities while maintaining convexity within each interval, we introduce the \emph{piecewise convex operator}
\begin{multline}
	\mathscr{C}(\mc{S})(\omega)
	:= \Bigl\{ (1-\lambda_\omega)S_k + \lambda_\omega S_{k+1} \;\Big|\;
	S_k \in \mathrm{co}(\Gamma_{\mc{S},k}^+),\\
	S_{k+1} \in \mathrm{co}(\Gamma_{\mc{S},{k+1}}^-) \Bigr\},
\end{multline}
for $\omega \in \mathcal{I}_{k|k+1}$ and $k \in \mathbb{I}_0^{N-2}$. At the boundary frequencies $\omega_k$, the values obtained from the left and right intervals may differ, i.e., $\mathscr{C}(\mc{S})(\omega_k^-) \neq \mathscr{C}(\mc{S})(\omega_k^+)$. This discontinuity is intentional and reflects that the generator sets $\Gamma_{\mc{S},k}^-$ and $\Gamma_{\mc{S},k}^+$ are not required to coincide—only that each interval remains convex. If the generator sets coincide at each sample, i.e., $\Gamma_{\mc{S},k}^- = \Gamma_{\mc{S},k}^+ := \Gamma_{\mc{S},k}$, the piecewise convex operator $\mathscr{C}(\mc{S})(\omega)$ reduces to the linear interpolation operator $\mathscr{L}(\mc{S})(\omega)$.

Next, recall the definition of $M$ in terms of \eqref{eq:CLasIQC} w.r.t. a given $P_\mr{zw}, \ldots,  P_\mr{yu}$ and $(V, U)$, the rcf parameterization of $K$. Then, $\widehat{M}_k$ can be similarly defined for  $\widehat{P}_{\mr{zw},k}, \ldots,  \widehat{P}_{\mr{yu},k}$ belonging to $\widehat{\mathcal{P}}_{\mr{zw},k}, \ldots,  \widehat{\mathcal{P}}_{\mr{yu},;}$ for a given single $(V, U)$, and such $\widehat{M}_k$ can be collected in  $\widehat{\mathcal{M}}_k$ similar to $ \widehat{\mathcal{P}}_k$. See Section \ref{Sec:constM} for the detailed construction. With these definitions in place, we are now ready to present the main theorem.
\begin{theorem}[Robust inter-sample synthesis]\label{thm:RigGuaran}
	Let $\ms{C}(\widehat{M})(\omega)$ denote the piecewise convex set defined by the generator sets $\Gamma_{M,k}^{-}$ and $\Gamma_{M,k}^{+}$ at each frequency point $\omega_k$ for $k \in \mathbb{I}_0^{N-1}$. Suppose the following inclusion holds:
	\begin{equation}\label{eq:DeltaI}
	\ms{P} \star K \subseteq \Delta_I \star \ms{M},
	\end{equation}
	where $\ms{P}$ is the uncertain plant class as defined in equation~\eqref{eq:uncertainGenPlantSystemClass}, $K$ is a given controller such that the interconnection $P \star K$ is well-posed for all $P \in \ms{P}$, and $\ms{M}$ is the uncertain system class defined as
	\begin{equation}\label{eq:RobM}
		\!\!\!\!\! \ms{M} = \left\{ M \in \mc{RH}_\infty^{\sbullet \times \sbullet}(\mb{E}) \;\middle|\; M(e^{j\omega}) \in \ms{C}(\widehat{M})(\omega) \right\},
	\end{equation}
	for $\omega \in \Omega$, such that the interconnection between every $M \in \ms{M}$ and every $\Delta_I \in \boldsymbol{\Delta}_I$ is well-posed. If
	\begin{itemize}
		\item there exist matrices $\Pi^\mathrm{s}_{k|k+1} \in \mathbb{H}_\mathrm{in}^{\sbullet \times \sbullet}$ and $\Pi^\mathrm{p}_{k|k+1} \in \mathbb{H}_\mathrm{in}^{\sbullet \times \sbullet}$ such that the following finite set of LMIs is satisfied:
		\begin{equation}\label{eq:finiteIQCFDIPerf}
			\Matrix{c}{\gamma_M \\ I}^*
			\Pi_{k|k+1} \Matrix{c}{\gamma_M \\ I} \prec 0,
		\end{equation}
		for all $\gamma_M \in \Gamma_{M,k}^+ \cup \Gamma_{M,{k+1}}^-$ with $k \in \mathbb{I}_0^{N-2}$ and
		$\Pi_{k|k+1} = \mr{daug}(\Pi_{k|k+1}^\mr{s}, \Pi_{k|k+1}^\mr{p})$;
		\item the \emph{zero-order-hold} (ZOH) operator applied to the collections $\{\Pi^\mathrm{s}_{k|k+1}\}_{k=0}^{N-2}$ and $\{\Pi^\mathrm{p}_{k|k+1}\}_{k=0}^{N-2}$, defined as
		\begin{equation*}
			\ms{Z}(\Pi)(\omega) := \Bigl\{ \Pi_{k|k+1} \;\Big|\; \forall \omega \in \mc{I}_{k|k+1}, \ k \in\mathbb{I}_0^{N-2}\Bigr\},
		\end{equation*}
		are such that $\ms{Z}(\Pi^\mathrm{s}) \in \boldsymbol{\Pi}_I$, with $\boldsymbol{\Pi}_I$ defined in~\eqref{eq:DeltaI}, and $\ms{Z}(\Pi^\mathrm{p}) \in \boldsymbol{\Pi}_\mathrm{p}$,
	\end{itemize}
	then the feedback interconnection $P \star K$ is robustly stable for all $P\in\ms{P}$, as certified by $\Pi^\mathrm{s}$, and meets the specified robust performance requirements, as characterized by $\Pi^\mathrm{p}$.
\end{theorem}
\begin{proof}
	Note that~\eqref{eq:finiteIQCFDIPerf} follows from multiplying 
	\begin{equation}\label{eq:combinedMult}
		\Pi_k := \Matrix{cccc}{\Pi^\mr{s}_{\mr{qq},k} & 0 & \Pi^\mr{s}_{\mr{qp},k} & 0 \\ 0 & \Pi^\mr{p}_{\mr{zz},k} & 0 & \Pi^\mr{p}_{\mr{zw},k} \\ (\Pi^\mr{s}_{\mr{zw},k})^* & 0 & \Pi^\mr{s}_{\mr{pp},k} & 0 \\ 0 & (\Pi^\mr{p}_{\mr{zw},k})^* & 0 & \Pi^\mr{p}_{\mr{ww},k}}.
	\end{equation}
	from the left with $[\widehat{M}_k^* \ I]$ and from the right with the conjugate transpose of $[\widehat{M}_k^* \ I]$ where $\widehat{M}_k$ is substituted by the generators in $\Gamma_{M,k}^{+} \cup \Gamma_{M,k+1}^{-}$.
	Note that $\Pi^\mr{s}_k,  \Pi^\mr{p}_k\in \mb{H}_\mr{in}^{\sbullet \times \sbullet} \Rightarrow \Pi_k \in \mb{H}_\mr{in}^{\sbullet \times \sbullet}$, and, by definition of $\ms{Z}$, we have that $\ms{Z}(\Pi)(\omega)\! \in\! \boldsymbol{\Pi}_\mr{in}^{\sbullet\! \times\! \sbullet}$. The mapping
	\begin{equation*}
         \Xi(\Pi_k, \gamma_M):=	
         \Matrix{c}{\gamma_M \\ I}^* \Pi_k \Matrix{c}{\gamma_M \\ I},
	\end{equation*}
	is convex for all $\gamma_M \in \mr{co}\left( \Gamma_{M,k}^+ \cup \Gamma_{M,{k+1}}^- \right)$ because $\Pi_{11, k} \succeq 0$. Negativity of $\Xi(\Pi, \gamma_M)$ at every generator $\gamma_M\in \Gamma_{M,k}\cup \Gamma_{M,{k+1}}$ implies negativity at all values $\gamma_M \in \mr{co}(\Gamma_{M_k}\cup \Gamma_{M_{k+1}})$. Now using the fact that $\ms{Z}(\Pi^ \mr{s})$ is an element of the set $\boldsymbol{\Pi}_I$, we have that
	\begin{equation}
		\Matrix{c}{\ms{M}(\omega) \\ I}^* \ms{Z}(\Pi)(\omega) \Matrix{c}{\ms{M}(\omega) \\ I} \prec 0, \ \forall\omega \in\partial\mb{D},
	\end{equation}
	implies robust stability for each $M \in \ms{M}$, with $\ms{M}$ as in~\eqref{eq:RobM}, which in turn implies that $P \star K$ is robustly stable and satisfies the performance specifications for every $P\in\ms{P}$.
\end{proof}

This result represents a fundamental advancement in the IQC-based analysis of systems with frequency-domain uncertainty. By employing convex uncertainty descriptions together with a finite-dimensional LMI characterization, the traditionally semi-infinite FDIs are rigorously reduced to a tractable finite-dimensional SDP. In contrast to existing approaches, this formulation enables exact certification of robust stability and performance directly from finitely many FRF measurements, while remaining convex in both the stability and performance multipliers. As such, it provides the first fully convex framework for frequency-domain robustness analysis based solely on finite, uncertain spectral data—without relying on state-space models or the KYP lemma.

Despite this progress, convexity is no longer guaranteed once the controller coprime factors $(V,U)$ depend on unknown design parameters $
\theta$, e.g., $V = V(\theta)$ and $U = U(\theta)$. In that case, the resulting constraints in~\eqref{eq:finiteIQCFDIPerf} become quadratic in $\theta$ and bilinear in the multipliers $\Pi$, breaking convexity of the synthesis problem. The following section addresses this issue by reformulating the problem setting to recover a finite-dimensional convex SDP that enables controller synthesis directly from uncertain, sampled FRF data, while providing guaranteed stability and performance over the entire frequency spectrum.

\subsection{Finite-Dimensional Convex Controller Synthesis} \label{sec:orb:synt}

In this subsection, we extend the finite-frequency IQC framework to controller synthesis. The objective is to formulate a convex optimization problem that determines the controller parameters directly from uncertain FRF data, while ensuring stability and performance guarantees across the entire frequency spectrum. To achieve this, we introduce a frequency-domain representation of the controller through its coprime factors $(V,U)$ and derive a convex reformulation that yields a finite-dimensional SDP suitable for synthesis. The forthcoming developments establish the precise parametrization, interpolation, and convex constraints required to obtain such controllers.

The main result will be converted to formally state robust stability and performance analysis in the form of an SDP:
\begin{equation}\label{eq:RobMultIQCite}
	\!\!\!\!\!\!	\begin{aligned}
		&\underset{\{\Pi^\mathrm{s}_{k|k+1}, \Pi^\mathrm{p}_{k|k+1}\}_{k=0}^{N-2} }{\mathrm{minimize}} & & \ \ \ \ \ J(\ms{Z}(\Pi^\mr{p})) \\
		&  \! \ \ \ \ \mr{subject \ to} & & \ \! \eqref{eq:finiteIQCFDIPerf} \ \mr{for \ all} \ \gamma_M \in \Gamma_{M,k}^+ \cup \Gamma_{M,{k+1}}^-\\
		& & & \ms{Z}(\Pi^\mr{s}) \in \boldsymbol{\Pi}_I, \ \ms{Z}(\Pi^\mr{p}) \in \boldsymbol{\Pi}_{\mr{p}}.
	\end{aligned}
\end{equation}
\begin{algorithm}[b]
	\caption{Iterative IQC-Based Controller Synthesis}\label{alg:robustAlg}
	\begin{algorithmic}[2]
		\Require Initial stabilizing controller $K_0 = U_0 V_0^{-1}$, given by $\theta_0\in\Theta$, generalized plant estimate generators $\Gamma^\mathrm{zw}_{P,k}, \ldots,
		\Gamma^\mathrm{yu}_{P,k}$, $k \in \mb{I}_0^{N-2}$, performance index $p \in \{2, \infty\}$
		\State {\bf construct} Generators $\Gamma_{M, k}^\pm$, see Section~\ref{Sec:constM}
		\State {\bf solve} SDP \eqref{eq:RobMultIQCite} under fixed $\theta_{0}$, giving $\Pi_{0}$
		\State {\bf set} $J^* \gets J(\ms{Z}(\Pi^\mr{p}))$ with $\Pi_{0},\theta_{0}$, $i\leftarrow 1$
		\Repeat
		\State {\bf solve} SDP \eqref{eq:RobContSDPAlg} under fixed $\Pi_{i-1}$, giving $\theta_{i}$
		\State {\bf solve} SDP \eqref{eq:RobMultIQCite} under fixed $\theta_{i}$, giving $\Pi_{i}$
		\State {\bf update}  $J^* \gets J(\Pi_{\mathcal{H}_p})$ with $\Pi_{i},\theta_{i}$
		\State {\bf update} controller $K \gets U_iV_i^{-1}$ based on $\theta_i$
		\State {\bf update} index $i \gets i+1$
		\Until{$J^*$ converged}
		\Ensure Controller $K$, ensuring robust stability and performance specifications 
	\end{algorithmic}
\end{algorithm}
Similar to the approach under exact and full spectral knowledge, treated in Section~\ref{sec:idealSol}, we aim to construct an equivalent SDP to~\eqref{eq:RobMultIQCite} which is convex in $\theta$ and $\{\Pi^\mr{p}_{k|k+1}\}_{k=0}^{N-2}$ for given $\{\Pi^\mr{s}_{k|k+1}\}_{k=0}^{N-2}$. Toward that end, we use the piecewise convexity on $\Omega$ of $\ms{C}(\widehat{M})(\omega)$ and the fact that the mapping $\Xi(\Pi_k, \gamma_M)$ is convex. By Schur's complement we have that
\begin{equation}\label{eq:RobContIQCite}
	\Matrix{cc}{\Pi_{22, k|k+1}+\mr{He}(\gamma_M^* \Pi_{12, k|k+1})  & \gamma_M^* \\ \gamma_M & -\Pi_{11, k|k+1}^{-1}} \prec 0,
\end{equation}
where $\mr{He}(X) := X^*+X$, is equivalent to~\eqref{eq:finiteIQCFDIPerf} on $\mc{I}_{k|k+1}$ for all generators $\gamma_M(\theta) \in \Gamma_{M,k}^+ \cup \Gamma_{M,{k+1}}^-$, all $\theta \in \Theta$, all $k\in\mb{I}_0^{N-2}$ and all $\Pi \in \boldsymbol{\Pi}_\mr{in}$. We note that~\eqref{eq:RobContIQCite} is indeed convex for in $\theta$ and $\{\Pi^\mr{p}_{k|k+1}\}_{k=0}^{N-2}$. Note that solving this LMI at the generators, negativity is ensured on the entire interval $\mc{I}_{k|k+1}$. Therefore, we can now state the following SDP
\begin{equation}\label{eq:RobContSDPAlg}
	\!\!\!\!\!\!	\begin{aligned}
		&\underset{\theta \in \Theta, \{\Pi^\mathrm{p}_{k|k+1}\}_{k=0}^{N-2} }{\mathrm{minimize}} & & \ \ \ \ \ J(\ms{Z}(\Pi^\mr{p})) \\
		&  \! \ \ \ \ \mr{subject \ to} & & \ \! \eqref{eq:RobContIQCite} \ \mr{for \ all} \ \gamma_M \in \Gamma_{M,k}^+ \cup \Gamma_{M,{k+1}}^-\\
		& & & \ms{Z}(\Pi^\mr{p}) \in \boldsymbol{\Pi}_{\mr{p}}.
	\end{aligned}
\end{equation}
The full finite-dimensional robust finite-data controller synthesis algorithm is presented in Algorithm~\ref{alg:robustAlg}. This algorithm provides a robust controller with stability and performance guarantees using only finite FRF data estimates of a generalized plant, which is a novel contributions to the field of convex FRF synthesis and the IQC framework as a whole.
\remark{Even though we tailored this section to the convex synthesis framework derived in \cite{bloemers2022frequency}, it should be noted that the main results of this section are of independent interest in general IQC-based analysis in the frequency domain since most results in literature do consider exact FRF data over the entire spectrum which was previously not computationally feasible without using the KYP lemma, e.g., having access to the underlying state-space model.}



%% file: sections/GeneratorConstruction.tex

\section{Generator Construction}\label{sec:genConstr}
%
In this section, we describe how to construct the generators $\Gamma_{M,k}^{\pm}$ required to formulate the synthesis problem is Section \ref{sec:orb:synt} as an SDP and represent the uncertainty of $\widehat{P}_k$ and subsequently $\widehat{M}_k$.
We first focus on generating covering sets for arbitrary unit-balls, which form the building blocks for representing uncertainty in dynamical systems and stochastic estimates. Once the unit-ball generators are defined, we discuss their application to more complex uncertainty descriptions.

\vspace{-2mm}
\subsection{Unit-Ball Generator Construction}

Unit balls are convex sets that serve as building blocks for more complex uncertainty descriptions. The matrix unit-ball with respect to a given norm $\|\sbullet\|$ is defined as
\begin{equation}
	\mathbb{B}(\mb{C}^{n \times m}, \|\sbullet\|) := \{ M \in \mb{C}^{n \times m} \mid \|M\| \le 1 \}.
\end{equation}
To construct generators for a unit-ball in $\mathbb{C}^{n \times m}$, we leverage the equivalence between the Frobenius norm and the Euclidean vector norm. Let the following bijective maps be defined:
\begin{equation}
	\begin{aligned}
		f_n&:\mathbb{C}^{n \times m} \rightarrow \mathbb{R}^{2n\times m}, & & f_n(M) = \begin{bmatrix} \mathrm{real}(M) \\ \mathrm{imag}(M) \end{bmatrix},\\
		g_n&:\mathbb{C}^{n \times m} \rightarrow \mathbb{C}^{nm}, & & g_n(M) = \mathrm{vec}(M),
	\end{aligned}
\end{equation}
which preserve the Frobenius norm, i.e., $\|M \|_\mathrm{FRO} = \|f_n(M)\|_\mathrm{FRO} = \|g_{2n}(f_n(M))\|_2$. Hence, these isomorphic maps allow us to construct generators for complex matrices from well-known constructions in real vector spaces. Let us denote the set of $n_\mathrm{g}$ independent generators for the unit-ball $\mathbb{B}(\mb{C}^{n \times m}, \| \sbullet \|)$ as $\mathbb{V}_{n_\mathrm{g}}(\mb{C}^{n \times m}, \|\sbullet\|)$ for which it is true that any $\Gamma\in\mathbb{V}_{n_\mathrm{g}}(\mb{C}^{n \times m}, \|\sbullet\|)$ satisfy that $\mathbb{B}(\mb{C}^{n \times m}, \| \sbullet \|) \subseteq \mr{co}(\Gamma)$.

Then, given any $\Gamma \in\mathbb{V}(\mathbb{R}^{2nm}, \|\sbullet \|_2)$, it is true that 
\begin{equation}
\left\{ (f_n^{-1}( g_{2n}^{-1}( \gamma)) \ \big| \ \gamma \in \Gamma 
\right\}= \breve \Gamma \in \mathbb{V}_{n_\mathrm{g}}(\mathbb{C}^{n\times m}, \|\sbullet\|_\mathrm{FRO}).
\end{equation}
Conversely, given any $\breve \Gamma \in \mathbb{V}_{n_\mathrm{g}}(\mathbb{C}^{n\times m}, \|\sbullet\|_\mathrm{FRO})$
\begin{equation}
\bigl\{ (g_{2n}( f_n( \breve\gamma)) \ \big| \ \breve\gamma \in \breve\Gamma 
\bigr\}=  \Gamma \in \mathbb{V}_{n_\mathrm{g}}(\mathbb{R}^{2nm}, \|\sbullet \|_2).
\end{equation}

To construct a generator set 
$\Gamma \in \mathbb{V}_{n_g}(\mathbb{R}^{2nm},\|\sbullet\|_2)$ 
that covers the Euclidean unit ball in $\mathbb{R}^d$, with $d = 2nm$,
a natural starting point is the case $n_g = d+1$, since no polytope with fewer
than $d+1$ vertices can cover a full-dimensional convex body in $\mathbb{R}^d$
(see, e.g.,~\cite{matousek2013lectures, schneider2013convex}).  
Among all polytopes in $\mathbb{R}^d$ with exactly $d+1$ vertices that
circumscribe the unit ball, the unique minimal-volume one 
(up to rotation) is the regular simplex embedding approach described in~\cite{berger2009geometry,coxeter1973regular}.

Following this construction,
consider the $(d+1)$ standard basis vectors 
$e_i \in \mathbb{R}^{d+1}$,
and define
\begin{equation}
	E_{d+1}
	=
	\left\{
	r^{-1} e_i
	\;\middle|\;
	i \in \mathbb{I}_1^{d+1}
	\right\},
	\qquad
	r := \frac{1}{\sqrt{d(d+1)}} ,
	\label{eq:simplex}
\end{equation}
so that $\mathrm{co}(E_{d+1})$ is a regular simplex of inradius~$1$.
Its barycenter is 
$c = r^{-1}(d+1)^{-1}\mathbf{1}_{d+1}$ and lies on the line orthogonal to the
$d$-dimensional hyperplane spanned by $E_{d+1}$.  The centered simplex
$V_{d+1} := E_{d+1} \ominus c$, with $\ominus$ the Minkowski difference, therefore spans a $d$-dimensional subspace of 
$\mathbb{R}^{d+1}$.  Any factorization
\begin{equation}\label{eq:qrDecomposition}
	\begin{aligned}
		V = AB, \ V = \Matrix{ccc}{v_1 & \cdots & v_{d+1}}, \ v_i \in V_{d+1}, i\in\mb{I}_{1}^{d+1}, \\
		A \in \mathbb{R}^{(d+1)\times d}, \ A^\top A = I_d, \ B
		\in \mathbb{R}^{d\times(d+1)},
	\end{aligned}
\end{equation}
projects the centered simplex onto $\mathbb{R}^{d}$ along an orthonormal basis
of its affine span.  The columns of $B$ then form a generator set
$\Gamma = \{b_i\}_{i=1}^{d+1}$ whose convex hull contains the unit ball with
equality of inradius.  In practice, such a decomposition is obtained directly
from a rank-revealing QR factorization of $V_{d+1}$. The construction is summarized in Algorithm~\ref{alg:simplexEmbedding}.

\begin{algorithm}[t]
	\caption{Generators for the simplex embedding}
	\label{alg:simplexEmbedding}
	\begin{algorithmic}[1]
		\Require $d \in \mathbb{N}$, $d \ge 1$
		\State {\bf set} $r^{-1} \gets \sqrt{d(d+1)}$
		\State {\bf set} $E_{d+1}$ using~\eqref{eq:simplex}
		\State {\bf set} $c \gets r^{-1}(d+1)^{-1}\mathbf{1}_{d+1}$
		\State {\bf set} $V_{d+1} \gets E_{d+1} \ominus c$
		\State {\bf compute} $A,B$ according to~\eqref{eq:qrDecomposition}
		\State {\bf set} $\Gamma \gets \{b_i\}_{i=1}^{d+1}$
		\Ensure $\Gamma$
	\end{algorithmic}
\end{algorithm}

The simplex embedding yields $n_g = d+1$ generators, which is the smallest
number of generators for which a polytope can cover the Euclidean unit ball in
$\mathbb{R}^d$.
Moreover, among all polytopes with $d+1$ vertices that circumscribe the unit
ball, the regular simplex uniquely minimizes volume up to orthogonal
transformations.
For $n_g > d+1$, determining the least conservative generator sets is
significantly more challenging: optimal polytopes are known only for certain
special values of $n_g$, and no general constructive characterization is
available.

\vspace{-2mm}
\subsection{Generator construction for generalized plants}
\subsubsection{Uncertain representations}
In this subsection, we will discuss uncertainty characterizations of the estimate $\widehat{P}$ of the generalized plant. Suppose that the available information about the system at a given frequency point 
$\omega_k \in \Omega_N$ is characterized by the estimate $\widehat{P}_k$ and is modeled 
using an additive uncertainty description\footnote{Input multiplicative uncertainty 
	$\widehat{P}(I+W_1 \Delta_\mathrm{u} W_2)$ and output multiplicative uncertainty 
	$(I+W_1 \Delta_\mathrm{u} W_2)\widehat{P}$ can be treated analogously to the additive formulation.}
\begin{multline} 
	\widehat{P}_k = \left\{ P_\mu(e^{j\omega_k}) + W_1(e^{j\omega_k}) \Delta_\mr{u}(e^{j\omega_k}) W_2(e^{j\omega_k}) \ | \right. \\ \left. \ \Delta_\mr{u} \in \boldsymbol{\Delta}_\mr{u} \right\}, \label{eq:adduncer}
\end{multline}
where $P_\mu$ is the nominal estimate of $P$, $W_1$ and $W_2$ are pre- and post-multiplication scaling matrices, respectively, and $\boldsymbol{\Delta}_\mr{u}$ is the unstructured dynamic uncertainty set:
\begin{equation} 
	\boldsymbol{\Delta}_\mr{u} := \left\{ \Delta_\mr{u} \in \mc{RH}_\infty(\mb{E})^{n_\mathrm{d} \times n_\mathrm{d}} \ | \ \|\Delta_\mr{u}\|_{\mc{H}_\infty} \leq 1 \right\}.
\end{equation}

We over-approximate the matrix $\Delta_\mr{u} (e^{j \omega_k})$ at every frequency point by noting that $\|\Delta_\mr{u}\|_{\mc{H}_\infty} \leq 1 \Rightarrow \|\Delta_\mr{u}(e^{j \omega})\|_{2,2} \leq 1$, $\forall \omega \in \Omega$. Hence, we have to construct generators for a convex cover of $\mathbb{B}(\mathbb{C}^{n\times m}, \|\sbullet\|_{2,2})$ with $n=n_\mathrm{d}$ and $m=n_\mathrm{d}$. We will accomplish this for the general case, when $n\neq m$.
Note that the following inequality holds for any $\Delta_\mr{u} \in \mb{C}^{n_\mathrm{d} \times n_\mathrm{d}}$
\begin{equation}
	\|\Delta_\mr{u}\|_{2,2} \leq \|\Delta_\mr{u}\|_\mr{FRO} \leq \sqrt{\min(n,m)} \|\Delta_\mr{u}\|_{2,2}.
\end{equation}
Using this, we can construct the required generators as:
\begin{equation}
\breve \Gamma \in \tfrac{1}{\sqrt{\min(n,m)}} \cdot  \mathbb{V}_{n_\mathrm{g}}(\mathbb{C}^{n\times m}, \|\sbullet\|_\mathrm{FRO}).
\end{equation}
Now let $\Gamma_\ast\in \mathbb{V}_{2nm+1}(\mathbb{R}^{2nm}, \|\sbullet \|_2)$ be the set of generators constructed by the simplex method. Then, because of the affine structure~\eqref{eq:adduncer}, the generators for the convex set $\widehat{P}_k$ are given by
\begin{multline}
	\Gamma_{P,k} = \bigl\{ P_\mu(e^{j\omega_k}) + W_1(e^{j\omega_k}) \gamma_{P} W_2(e^{j\omega_k}) \mid \\   \gamma_P \in \tfrac{1}{\sqrt{\min(n,m)}} f_n^{-1}(g_{2n}^{-1}(\Gamma_\ast))\bigr\}.
\end{multline}
These generators are directly applicable to the main result of the previous section.
\subsubsection{Stochastic representations} \label{sec:stoc:gen}


Often when the FRF estimate of a plant is obtained, instead of the bounded uncertainty characterization in \eqref{eq:adduncer}, only the mean and covariance of  $P(e^{j\omega_k})$ are available. Hence, we will consider how generators can be obtained for such a statistical formulation of the uncertainty in this subsection.  Characterization of the generators will be obtained in terms of a vectorized multivariate normal distribution, for which we will assume that estimates at different frequency points are statistically independent and that the noise contributions are circular complex normally distributed~\cite{pintelon2012system}.

Let the FRF estimate $\widehat{P}_k := \widehat{P}(e^{j\omega_k}) \in \mb{C}^{n \times m}$ be given, at a certain frequency point $\omega_k \in \Omega_N$, as the complex circular symmetric multivariate normal distribution
\begin{equation}
	\textbf{v}_k \sim \mc{N}_{\mb{C}^{n m}} \left( \mu_k, \Sigma_k \right),
\end{equation}
with $\textbf{v}_k$ a $nm$-dimensional random complex circularly symmetrical variable, corresponding to $\mathrm{vec}(\widehat{P}_k)$, with mean $\mu_k \in \mb{C}^{n m}$  and covariance matrix $\Sigma_k \in \mb{H}_{\succ 0}^{nm}$. We can represent the multivariate complex normal distribution as
\begin{equation}
	\textbf{v}_k = \mu_k + S_k  \frac{\textbf{a}+j \textbf{b}}{\sqrt{2}}, \quad \textbf{a},\textbf{b} \sim \mc{N}_{n m}(0, I), \quad \textbf{a}\perp \textbf{b},
\end{equation}
where $S_k$ is any square matrix such that $S_kS_k^* = \Sigma_k$ and $\textbf{a} \perp \textbf{b}$ denotes that $\textbf{a}$ and $\textbf{b}$ are statistically independent. The random variables $\textbf{a}$ and $\textbf{b}$ are independent of $k$ as well. Let $\textbf{r}_k := f_{nm}(\textbf{v}_k)$ be the real embedding of the complex variable $\textbf{v}_k$, for which we obtain that
\begin{equation}
	\textbf{r}_k = f_{nm}(\mu_k) + L_k \textbf{x}, \quad \textbf{x} = \Matrix{cc}{\textbf{a}\\ \textbf{b}} \sim \mc{N}_{2nm}(0,I),
\end{equation}
with
\begin{equation*}
	L_k = \frac{1}{\sqrt{2}}\Matrix{cc}{\mr{real}(S_k) & - \mr{imag}(S_k) \\ \mr{imag}(S_k) & \mr{real}(S_k)}.
\end{equation*}
The normal multivariate distribution of $\textbf{r}_k$ is given by
\begin{equation}
	\textbf{r}_k \sim \mc{N}_{2nm}\left(f_{nm}(\mu_k), \frac{1}{2} \Matrix{cc}{\mr{real}(\Sigma_k) & \mr{imag}(\Sigma_k) \\ -\mr{imag}(\Sigma_k) & \mr{real}(\Sigma_k)} \right).
\end{equation}
Given a confidence $\alpha \in [0, 1)$, the corresponding set
\begin{equation}
	\left\{ x\in\mathbb{R}^{2nm} \mid \|\textbf{x}-x\|_2 \leq r(\alpha),\  \textbf{x}\sim \mc{N}_{2nm}(0, I) \right\},
\end{equation}
is the ball  $r(\alpha)\cdot \mathbb{B}(\mathbb{C}^{2nm}, \|\sbullet \|_2)$, hence
\begin{equation}
	C_k(\alpha) := \mu_k \oplus r(\alpha) \ L_k\cdot \mathbb{B}(\mathbb{R}^{2nm}, \|\sbullet \|_2),
\end{equation}
gives the convex set in which $\textbf{r}_k$ falls with probability $\alpha$. The radius $r(\alpha)$ is equal to the square root of the inverse cumulative distribution function, of the Chi-squared distribution $\chi^2_{2nm}$ with $2nm$ degrees of freedom, denoted by $\sqrt{Q_{\chi_{2nm}^2}(\alpha)}$.

Now let $\Gamma_\ast\in \mathbb{V}_{2nm+1}(\mathbb{R}^{2nm}, \|\sbullet \|_2)$ be the set of generators constructed by the simplex method.  The generators that cover $C_k(\alpha)$ are given by
\begin{multline}
	\Gamma_{C_k(\alpha)} = \bigl\{  f_{nm}(\mu_k) +  r(\alpha)L_k\gamma \mid   \gamma \in \Gamma_\ast \bigr\},
\end{multline}
from which it follows immediately that the generators that cover $\widehat{P}_k$, with probability $\alpha$, is
\begin{equation}
	\Gamma_{P_k, \alpha} = \bigl\{  g^{-1}_n\left( f^{-1}_{nm}\left(\gamma \right) \right) \mid   \gamma \in \Gamma_{C_k(\alpha)} \bigr\}.
\end{equation}

\subsection{Constructing $\Gamma_{\widehat{M}_k}^{\pm}$} \label{Sec:constM}
In this section, we provide construction for the generators $\Gamma_{M_k}^-$ and $\Gamma_{M_k}^+$  that characterize the piecewise convex set $\ms{M}$. This process is divided into three steps.
\vspace{0mm}
\subsubsection{General consideration}
\textcolor{white}{ }\\
The first step is to characterize the uncertain interconnection $\ms{P} \star K$ on each interval $\mc{I}_{k|k+1}$ for which we note that it holds that
\begin{equation}
	\ms{P}(\omega) \star K(e^{j\omega}) \subseteq \ms{L}(\widehat{\mc{P}})(\omega) \star K(e^{j\omega}).
\end{equation}
For brevity of the notation, we will use $K_\omega, \ldots, V_\omega$ instead of $K(e^{j\omega}), \ldots, V(e^{j\omega})$. We can write the LFT $\ms{P}(\omega) \star K_\omega$ using $K_\omega = U_\omega V_\omega^{-1}$ as
\begin{equation*}
	\ms{P}_\mr{zw}(\omega) + \ms{P}_\mr{zu}(\omega) U_\omega V_\omega^{-1} \left( I - \ms{P}_\mr{yu}(\omega) U_\omega V_\omega^{-1}\right)^{-1} \ms{P}_\mr{yw}(\omega),
\end{equation*}
because the linear interpolation operator preserves the blockwise structure of $\widehat{\mc{P}}$. Next, we decompose $\ms{L}(\widehat{\mc{P}}_\mr{yu})(\omega)$ as:
\begin{equation}\label{eq:uncerCoprimeG}
	\ms{L}(\widehat{\mc{P}}_\mr{yu})(\omega) := D_{k+1/2}^{-1} N_{k+1/2} + \ms{L}(\overline{\mc{P}}_\mr{yu})(\omega),
\end{equation}
on each $\mc{I}_{k|k+1}$, $k\in \mb{I}_0^{N-2}$ such that $(D_{k+1/2}, N_{k+1/2})$ denotes the sampled version of a lcf $(D, N)$ for which it holds that $D_{k+1/2}^{-1} N_{k+1/2} = \frac{1}{2}(c(\Gamma_{\mr{yu}, k}) + c(\Gamma_{\mr{yu}, k+1}))$. 

The generators of $\ms{L}(\overline{\mc{P}}_\mr{yu})(\omega)$ at $\omega_k$ and $\omega_{k+1}$ simply are $\widehat{\mc{P}}_{\mr{yu},k}-D_{k+1/2}^{-1}N_{k+1/2}$ and $\widehat{\mc{P}}_{\mr{yu},k+1}-D_{k+1/2}^{-1}N_{k+1/2}$, respectively. We use the same LFR trick from the previous section to pull out $\Delta_I$ and, in the process, we obtain a characterization of the  family of sets denoted by $\ms{F}(\omega)$:
\begin{equation}
	\ms{F}(\omega) = \Matrix{cc}{I-\Lambda(\omega) & \  D_k \ms{L}(\widehat{\mc{P}}_\mr{yw})(\omega)  \\ \ms{L}(\widehat{\mc{P}}_\mr{zu})(\omega) U_\omega & \  \ms{L}(\widehat{\mc{P}}_\mr{zw})(\omega)},
\end{equation}
where
\begin{equation}
	\Lambda(\omega) = D_{k+1/2} ( V_\omega- \ms{L}(\overline{\mc{P}}_\mr{yu})(\omega) U_\omega) - N_{k+1/2} U_\omega,
\end{equation}
for $\omega \in \omega \in \mc{I}_{k|k+1}$. We note that $\ms{L}(\widehat{\mc{P}})(\omega) \star K_\omega = \Delta_I \star \ms{F}(\omega)$. Now we need to over approximate $\ms{F}(\omega)$ by a piece-wise convex set $\ms{M}(\omega)$ with finite generators. First, we note that the presence of $V_\omega$ and $U_\omega$ prevent convexity over the interval, while it does not destroy convexity at a fixed $\omega$. 

\subsubsection{Piece-wise convex sets of controller coprime factors}
\textcolor{white}{ }\\
Next, piece-wise convex regions $\ms{U}(\omega)$ and $\ms{V}(\omega)$ are defined such that $U_\omega \in \ms{U}(\omega)$ and $V_\omega \in \ms{V}(\omega)$, respectively. The uncertainty, due to sampling, in both $V_\omega$ and $U_\omega$ arises in $\phi_\omega$. The controller parameters $(V, U)$ are defined for all $\omega$ but to make the approach computational tractable we can only use the sampled version, e.g., $(V_k, U_k)$ for $k = 0, 1, \cdots, N-1$. In order to characterize a piece-wise convex set $\Phi(\omega)$ such that $\phi_\omega \in \Phi(\omega)$ let us first define the following $N-1$ scalars
\begin{equation}\label{eq:controllerUncertainty}
	r_{k|k+1} = \sup_{\omega \in \mc{I}_{k|k+1}} \| \ms{L}(\phi)(\omega)-\phi_\omega \|_\mr{FRO}.
\end{equation}
These scalars can be approximated by sufficiently dense sampling offline. Here, $\ms{L}(\phi)$ is the linear interpolated curve between the points $\phi_{\omega_k}$. With these scalars we can define the following piece-wise convex set:
\begin{equation}
	\begin{aligned}
	\Phi(\omega) := \left\{ \ms{L}(\phi)(\omega) + r_{k|k+1} \Delta_\phi \ | \ \Delta_\phi \in  \mr{co}(\Gamma_\phi) \right\}.
	\end{aligned}
\end{equation}
An equivalent characterization is $\Phi(\omega) = \ms{C}(\widehat{\Phi})(\omega)$ where $\widehat{\Phi}_k$ is characterized by the generators $\Gamma_{\Phi_k}^- := \phi_{\omega_k}+r_{k-1|k} \Gamma_\phi$ and $\Gamma_{\Phi_k}^+ := \phi_{\omega_k}+r_{k|k+1} \Gamma_\phi$. Based on this piecewise convex set, we define the piecewise convex sets of the coprime factors as
\begin{equation}
	\begin{split}
		\ms{V}(\omega) &= C_V \Phi(\omega)+D_V,\\
		\ms{U}(\omega) &= C_U \Phi(\omega)+D_U.
	\end{split}
\end{equation}
Now that we have polytopic regions defined, we note that there is a bilinear product of such convex sets still present in $\ms{M}$ for which we have to define a convex polytopic over approximation. This is done in the next section.
\vspace{0mm}
\subsubsection{Polytopic over approximation} 
\textcolor{white}{ }\\
By substituting the piecewise convex sets $\ms{V}(\omega)$ and $\ms{U}(\omega)$ of $V_\omega$ and $U_\omega$, respectively, into $\ms{F}(\omega)$ we get bilinear products of two piecewise convex sets on $\mc{I}_{k|k+1}$. The bilinear product of two piecewise convex sets is not a piecewise convex set itself.
In this section, we provide a polytopic over-approximation of a bilinear product of two piece-wise convex sets which is used to find an polytopic over-approximation of $\ms{M}(\omega)$.  

Towards this end, let $\mc{A}(\omega)$ and $\mc{B}(\omega)$ be two piecewise convex sets given by $\ms{C}(\widehat{A})(\omega)$ with generators $\Gamma_{A,k}^\pm$ and $\ms{C}(\widehat{B})(\omega)$ with generators $\Gamma_{B,k}^\pm$, respectively. Then, let the bilinear product of the two piecewise convex sets be given by
\begin{equation}
	\!\!\!\! \mc{C}(\omega) = \left\{ A_\omega B_\omega \ \mid \ A_\omega \in \ms{C}(\widehat{A})(\omega), \ B_\omega \in \ms{C}(\widehat{B})(\omega) \right\}.
\end{equation}
We will use the following notation in order to obtain a piecewise convex over approximation of $\mc{C}(\omega)$
\begin{equation}
	\ms{C}_\otimes(\widehat{A}, \widehat{B}) (\omega)  := \ms{C}(\widehat{C})(\omega),
\end{equation}
where $\widehat{C}$ is defined by the generators
\begin{multline}
	\Gamma_{A, B, k|k+1} = \{ \gamma_A \gamma_B \ \mid \ \gamma_A \in \Gamma_{A,k}^+ \cup \Gamma_{A,k+1}^- , \\ \gamma_B \in \Gamma_{B,k}^+ \cup \Gamma_{B,k+1}^- \},
\end{multline}
such that $\Gamma_{C,k|k+1} = \Gamma_{C,k}^+ = \Gamma_{C,k+1}^-$. With this notation and these over-approximations in place, we now present the generators of $\widehat{M}_k$. First, we provide generators $\Gamma_\Lambda^\pm$ that specify an over-approximation of $\Lambda(\omega) \subseteq \ms{C}(\widehat{\Lambda})(\omega)$ which are:
\begin{subequations}
	\begin{align}
	\!\!\!\! \Gamma_{\Lambda, k}^+ &= D_{k|k+1} \left( \Gamma_{V,k}^+ -\Gamma_{\overline{P}_\mr{yu}, U, k|k+1} \right) - N_{k|k+1} \Gamma_{U, k}^+,\\
	\!\!\!\! \Gamma_{\Lambda, k}^- &= D_{k-1|k} \left( \Gamma_{V,k}^- -\Gamma_{\overline{P}_\mr{yu}, U, k-1|k} \right) - N_{k-1|k} \Gamma_{U, k}^-.
	\end{align}
\end{subequations}
And, hence, the generators of $\widehat{M}_k^\pm$ are given by
\begin{subequations}
	\begin{align}
	\Gamma_{M, k}^+ &= \Matrix{cc}{I - \Gamma_{\Lambda, k}^+ & D_{k|k+1} \Gamma_\mr{yu} \\ \Gamma_{\ms{L}(P_{zu}), U, k|k+1} & \Gamma_\mr{zw}},\\
	\Gamma_{M, k}^- &= \Matrix{cc}{I - \Gamma_{\Lambda, k}^- & D_{k-1|k} \Gamma_\mr{yu} \\ \Gamma_{\ms{L}(P_{zu}), U, k-1|k} & \Gamma_\mr{zw}}.
\end{align}
\end{subequations}
Lastly, we conclude that
\begin{equation*}
	\ms{P}(\omega) \star K \subseteq \Delta_I \star \ms{C}(\widehat{M}) (\omega),
\end{equation*}
providing an explicit construction of the generator $\Gamma\pm_{M,k}$ required in Theorem~\ref{thm:RigGuaran} and the Synthesis Algorithm~\ref{alg:robustAlg} rendering the main contribution of this paper applicable.

\subsection{Discussion}
Up untill this point, the exposition is done at the generalized plant level. In most practical situations, however, the generalized plant $P$ is a function of the to-be-controlled plant $G$. A popular example of this is the 4-Block control configuration. Associated to this control configurations, the weighted cost function is given by
\begin{equation}
	\left\| W_\mr{z} \Matrix{cc}{S & S G \\ K S & KSG} W_\mr{w} \right\|_{\mc{H}_p}, \ S = (I+GK)^{-1}.
\end{equation}
The corresponding unweighted generalized plant $P$ is then given by
\begin{equation}
	P = \Matrix{ccc}{I & G & -G \\ 0 & -I & I \\ I & G & -G},
\end{equation}
for which we note that
\begin{equation}
	\left(\Matrix{cc}{W_z & 0 \\ 0 & I} P \Matrix{cc}{W_w & 0 \\ 0 & I} \right) \star K = \Delta_I \star \left(\Matrix{cc}{I & 0 \\ 0 & W_z} M \Matrix{cc}{I & 0 \\ 0 & W_w} \right).
\end{equation}
Hence, $M$ can be constructed from the unweighted generalized plant and, in most situations, it holds true that $P_\mr{yu} = -G$. Now, based on finite estimate $\Gamma_{\widehat{G},k}$ we construct sampled coprime factors for every interval $[\omega_k, \omega_{k+1}]$. Specifically, let
\begin{equation}
	G_{\mu,k+1/2} :=\frac{c(\Gamma_{\widehat{G},k})+c(\Gamma_{\widehat{G},k+1})}{2},
\end{equation}
where $c(\cdot)$ is defined in Appendix~\ref{app:convex}. Moreover, let $\omega_{k+1/2} = \frac{\omega_k+\omega_{k+1}}{2}$, then we define the sampled coprime factorization by
\begin{equation}
	\begin{split}
		D_{k+1/2} &= (I+G_{\mu,k+1/2} K_0(e^{j \omega_{k+1/2}}))^{-1},\\ 
		N_{k+1/2}& = (I+G_{\mu,k+1/2} K_0(e^{j \omega_{k+1/2}}))^{-1} G_{\mu,k+1/2},
	\end{split}
\end{equation}
where $K_0 \in \mc{RH}_{\infty}^{n_\mr{u} \times n_\mr{y}}$ represents any stabilizing controller. Now, the generators $\widehat{M}_k$ can be constructed as shown in the previous subsections. It must also be noted that the simplex embedding approach for generator construction provides the least amount of generators possible, which might be conservative. There are advanced sampling strategies that provide less conservative generators, this however, goes at the expense of computational requirements.

Lastly, for simplicity, we omit the inter-frequency behavior of the weighting filters since this is negligable for low-order weighting filters. This could of course be incorporated using the polytopic overapproximation of the bilinear product of two convex sets as described in the previous subsection.

%% file: sections/Wirebonder.tex
\section{Industrial Wirebonder Application}
\label{sec:example}

This section demonstrates the proposed finite-data IQC-based robust controller synthesis approach on a real industrial system. The objective is to validate the methodology on a practical high-performance MIMO control problem using experimentally measured frequency response data, without relying on a parametric or physics-based model of the plant.

\begin{figure}[t]
	\centering
	\includegraphics[width=0.3\linewidth]{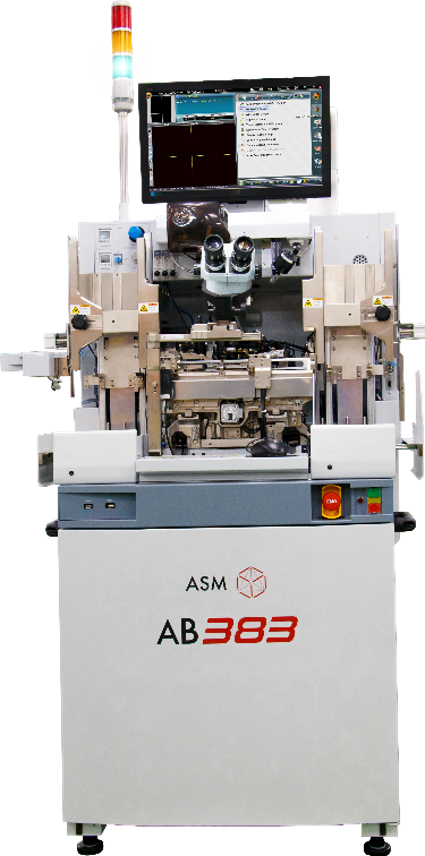}
	\hfill
	\includegraphics[width=0.65\linewidth]{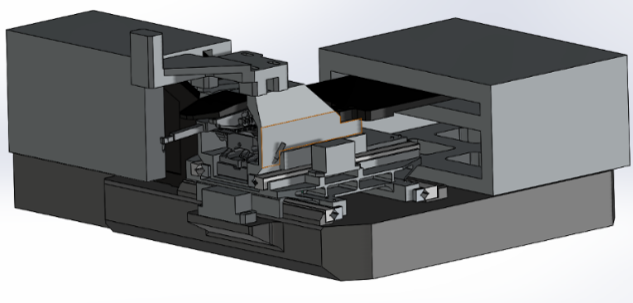}
	\caption{Industrial wirebonder setup (ASM AB383) used for experimental validation shown on the left and schematics for a motion stage on the right.}
	\label{fig:commercWirebonder}
\end{figure}

\subsection{System Description}\label{mustwork}


The experimental setup consists of a commercial ASM AB383 wirebonder, shown on the left side in  Figure~\ref{fig:commercWirebonder}. Wirebonders are high-precision mechatronic systems used in semiconductor manufacturing to establish electrical connections between integrated circuits and their packaging. The machine features a stacked planar XY motion stage actuated by linear 3-phase motors and equipped with high-resolution encoders for position feedback. A real-time control interface enables closed-loop operation, data acquisition, and controller deployment. Schematics of the motion stageare shown on the right side in Figure~\ref{fig:commercWirebonder}. The associated lab setup, including cabling and cooling, is shown in Figure~\ref{fig:wirebonder_lab}.

From a control perspective, the wirebonder constitutes a strongly coupled MIMO system with significant flexible dynamics in the frequency range relevant for high-speed bonding operations. These dynamics originate from structural compliance, actuator–structure interactions, and cross-axis coupling. As a result, constructing an accurate parametric model suitable for high-performance controller synthesis is challenging. Hence, the system is characterized exclusively through an experimentally measured FRF, which directly capture the relevant input–output behavior, including flexible modes, coupling effects, and the associated uncertainty. These FRF measurements serve as the sole plant description used in the subsequent identification, uncertainty modeling, and controller synthesis steps.
Figure~\ref{fig:wirebonder_lab} shows the laboratory setup to which the techniques developed in this paper are applied. This setup serves as the experimental platform for the controller synthesis and validation.
\begin{figure}[!t]
	\centering
	\includegraphics[width=0.75\linewidth]{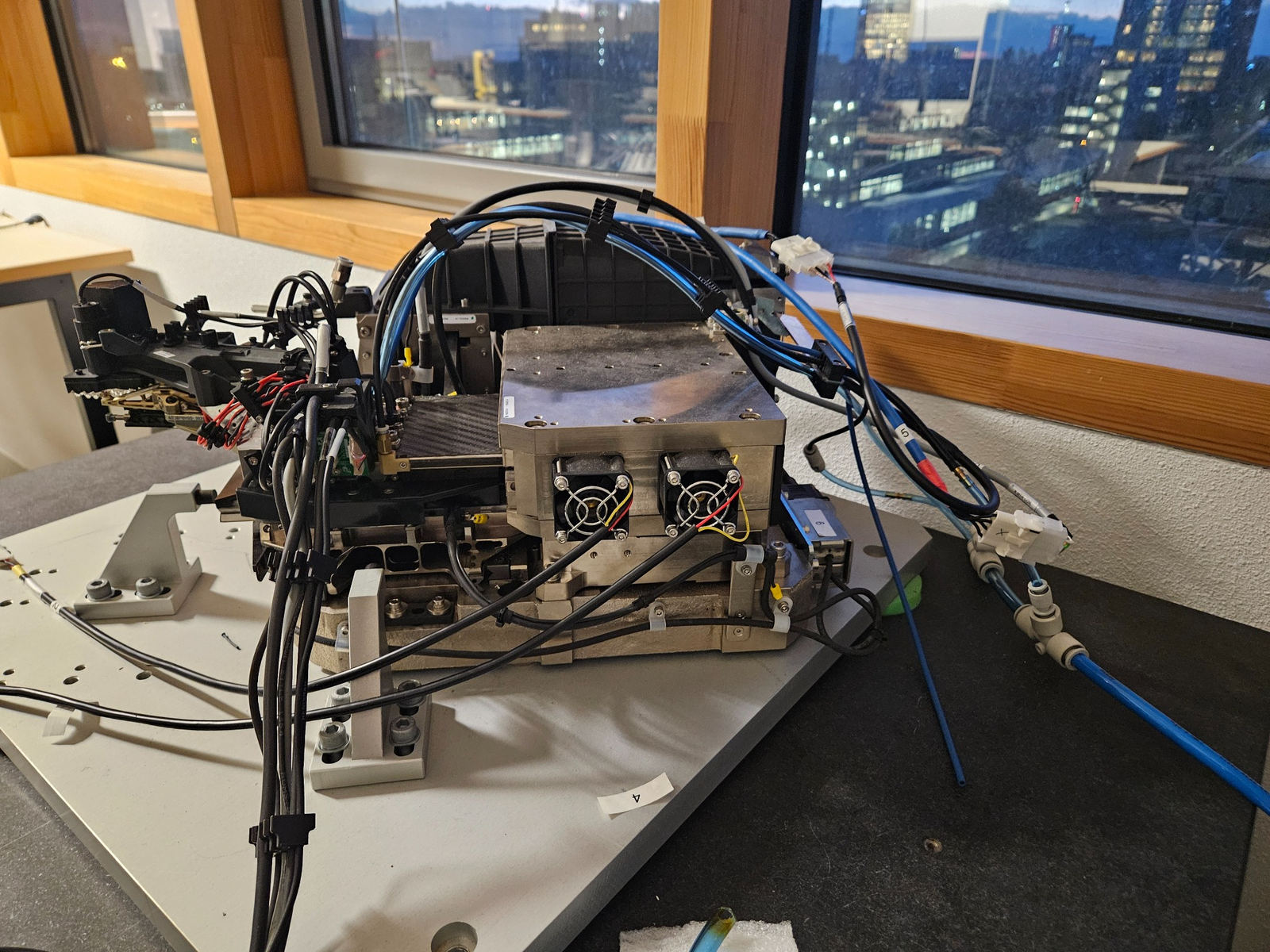}
	\caption{Motion stage of the wirebonder.}
	\label{fig:wirebonder_lab}
\end{figure}

\subsection{Identification and generator construction}\label{mustwork2}

The FRF of the motion stage of the wirebonder is obtained using orthogonal multisine excitation~\cite{dobrowiecki2006optimized} while operating the setup under a hand-tuned diagonal lead controller for stabilization. The multisine excites $247$ logarithmically spaced frequencies between the fundamental frequency $f_0$ and the Nyquist frequency $f_s/2$, rounded to the nearest integer multiple of $f_0$ to ensure periodic excitation. The plant inputs and outputs are subsequently processed using the \texttt{RobustLocalPolyAnal.m} routine from the Frequency Domain Toolbox~\cite{pintelon2012system}, implementing the Local Polynomial Method~\cite{gevers2011local} for nonparametric frequency-response estimation. Ten independent multisine realizations are generated, each repeated ten times and averaged within the identification framework. The exact sequence length and sampling frequency are omitted for confidentiality reasons. The resulting FRF is linearly interpolated onto a logarithmically spaced grid of $740$ frequency points spanning the same frequency range, which is used throughout the remainder of the paper for visualization and controller synthesis.

The resulting FRF estimate is characterized by the mean
$G_\mu(e^{j\omega_k})\in\mathbb{C}^{n_y\times n_u}$
and covariance
$\Sigma_G(e^{j\omega_k})\in\mb{H}_{\succ0}^{n_y n_u}$
at each measured frequency $\omega_k\in\Omega_N$, where $N=247$. The generator sets $\Gamma_{G_k}$ are constructed according to the procedure in Subsection~\ref{sec:stoc:gen} using a confidence region with probability of approximately $95.45\%$, corresponding to a $2\sigma$ confidence bound. The uncertain FRF estimate is shown in Fig.~\ref{fig:uncerBode}, where the mean estimate is depicted together with shaded confidence region, corresponding to the uncertainty bound. The off-diagonal FRFs exhibit relatively large uncertainty at low frequencies, whereas the diagonal dynamics become predominantly uncertain in the high-frequency region.

We design a controller by minimizing the weighted sensitivity objective $	\left\|W_\mr{z} (I+GK)^{-1}\right\|_{\mc{H}_\infty}$ where the weighting filter $W_\mr{z}$ encodes the desired sensitivity shape. The estimate of the generalized plant, based on the generators $\Gamma_{G_k}$ describing the estimate of $G$, is given by
\begin{equation}
	\widehat{\mc{P}}_k = \Matrix{cc}{I & -\mr{co}(\Gamma_{G_k}) \\ W_\mr{z} & -W_\mr{z} \mr{co}(\Gamma_{G_k})}.
\end{equation}
Next, generators for the coprime factors of the estimated FRF are constructed for every interval $[\omega_k, \omega_{k+1}]$ as specified by~\eqref{eq:uncerCoprimeG}. Specifically, let
\begin{equation}
	G_{\mu,k|k+1} :=\frac{G_\mu(e^{j \omega_k})+G_\mu(e^{j \omega_{k+1}})}{2},
\end{equation}
and let $\omega_{k|k+1} = \frac{\omega_k+\omega_{k+1}}{2}$, then we define the sampled, nominal coprime factorization by
\begin{equation}
	\begin{split}
		D_{k|k+1} &= (I+G_{\mu,k|k+1} K_0(e^{j \omega_{k|k+1}}))^{-1},\\ 
		N_{k|k+1}& = (I+G_{\mu,k|k+1} K_0(e^{j \omega_{k|k+1}}))^{-1} G_{\mu,k|k+1},
	\end{split}
\end{equation}
where $K_0 \in \mc{RH}_{\infty}^{n_\mr{u} \times n_\mr{y}}$ can in general be any stabilizing controller; in particular, here the controller $K_0$ is taken as the implemented controller for closed-loop identification. With these sampled, nominal coprime factorization available, the generators $\Gamma_{M_k}^\pm$ can be constructed as shown in Subsection~\ref{Sec:constM}.
\begin{figure}[t]
	\centering
	\includegraphics[width=0.98\linewidth]{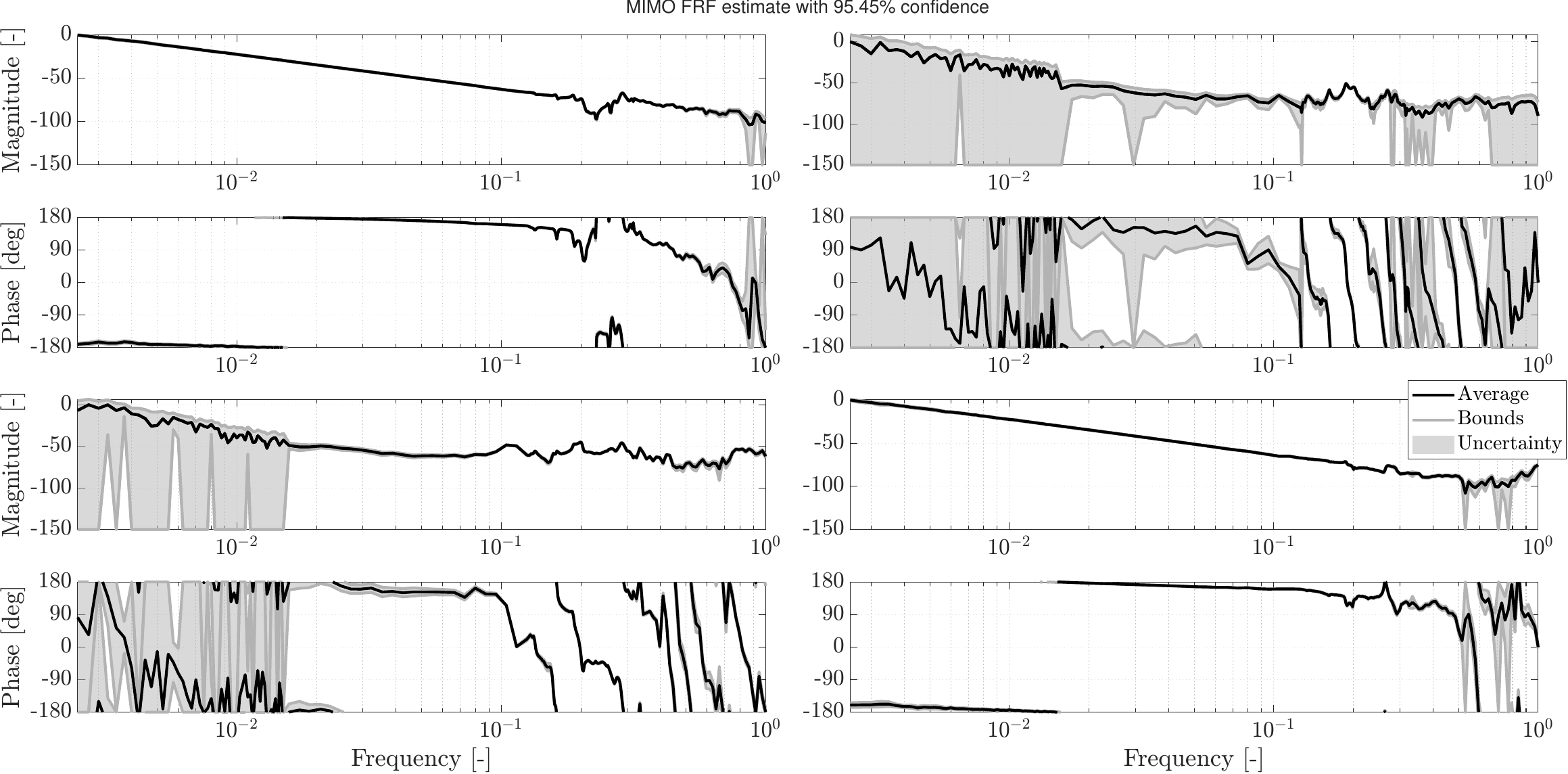}
	\caption{Bode plots of the estimated FRF of the MIMO wirebonder together with 95.45\% confidence bounds.}
	\label{fig:uncerBode}
\end{figure}

\subsection{Controller design}\label{mustwork3}
The robust controller synthesis Algorithm \ref{alg:robustAlg} requires an initially stabilizing controller $K_0 = U_0 V_0^{-1}$. Here, we use the same controller for the initialization as was used for closed-loop system identification for which normalized left coprime factorization \cite{zhou1996robust} is used to determine $U_0$ and $V_0$. To ensure that the controller parametrization does not introduce excessive variability relative to the plant uncertainty description, the controller order $n_\phi$ is chosen significantly smaller than the number of frequency points of the plant estimate. In particular, $n_\phi$ is selected to be at least twenty times smaller than the frequency grid size, which acts as a form of regularization and avoids spurious high-frequency behavior induced by the coprime factor parametrization. This design choice is consistent with standard robust control practice favoring low-order controllers relative to plant complexity \cite{skogestad2007multivariable}. This also makes sure that $r_{k|k+1}$ in Equation~\eqref{eq:controllerUncertainty} is relatively small. Therefore, we fix the controller order $n_\phi$ to be 12. The GOBFs used in the controller coprime factorization $\phi$ are chosen as
\begin{equation}
	\phi(z) = \Matrix{cccccccc}{1 & a_1(z) & \cdots & a_6(z) & 0 & 0 & \cdots & 0 \\ 0 & 0 & \cdots & 0 &  1 & a_1(z) & \cdots & a_6(z)}^\top,
\end{equation}
where $a_i(z) = \frac{n_i(z)}{(z-0.96)^i}$ and the numerators $n_i$ are such that all 7 basis function are orthogonal in $\mc{RH}_2(\mb{E})$, see~\cite{van2000modelling}. The frequency associated with the pole location $p = 0.96$ is chosen to reflect the desired bandwidth. In order to realize a state-space controller, as shown in Equation~\eqref{eq:controllerRealization}, the matrices $A_\phi$ and $B_\phi$ are determined from a balanced realization of the GOBFs $\phi(z)$.

In order to force integral action and high-frequency roll-off of the controller, the following parametrization of the feedthrough matrices $D_V$ and $D_U$ of the coprime factors $V$ and $U$, are applied:
\begin{subequations}
	\begin{align}
		D_V &= - C_V \phi(1), \\
		D_U &= - C_U \phi(-1).
	\end{align}
\end{subequations}
Hence, the optimization variables are reduced to $\theta = \mr{vec}(C_V, C_U)$. 

The diagonal weighting filter $W$ employed to shape the desired mixed-sensitivity specification is shown in Figure~\ref{fig:uncerSens}, together with the resulting closed-loop sensitivity function. The weighting filter encodes frequency-dependent performance objectives by emphasizing disturbance attenuation and tracking performance in the low-frequency range, while relaxing constraints at higher frequencies where robustness considerations dominate. 

In Figure~\ref{fig:uncerSens} the achieved normalized closed-loop sensitivity, evaluated over the uncertain plant FRF, is shown. The results demonstrate that the synthesized controller successfully enforces the prescribed sensitivity shaping across the frequency range of interest, despite the presence of significant plant uncertainty and multivariable coupling. In particular, the closed-loop sensitivity remains uniformly bounded by the inverse of the weighting filter, indicating that the mixed-sensitivity objectives are satisfied for all sampled uncertainty realizations.

\begin{figure}[t]
	\centering
	\includegraphics[width=0.98\linewidth]{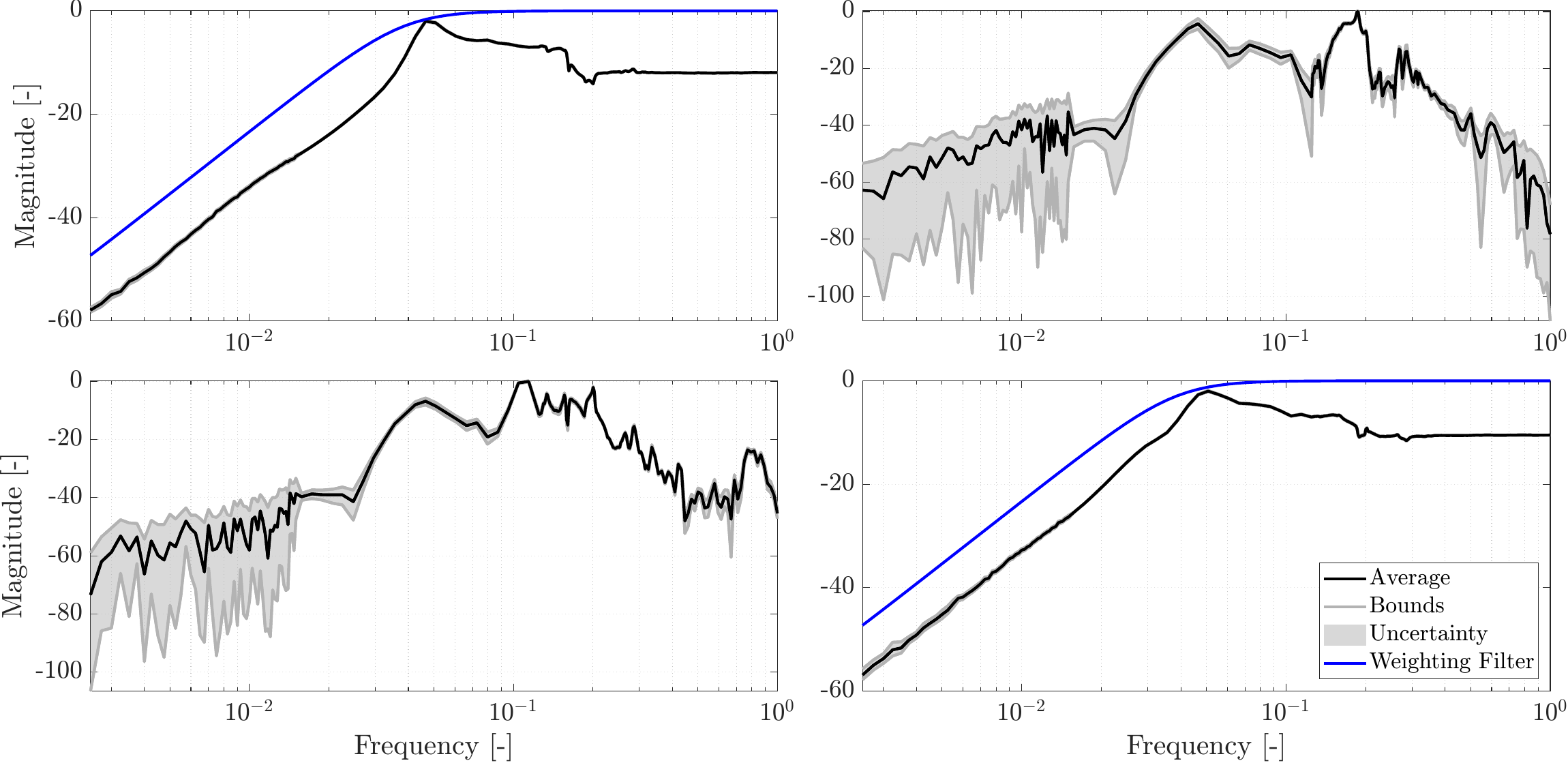}
	\caption{Shaped sensitivity function of the closed-loop calculated uncertain FRF estimate of the wirebonder and the designed controller.}
	\label{fig:uncerSens}
\end{figure}

These results highlight the effectiveness of the proposed frequency-domain, data-driven synthesis approach in translating frequency performance specifications, expressed via weighting filters, into verifiable closed-loop behavior on an uncertain MIMO system.


\subsection{Experimental validation}\label{mustwork4}
To experimentally validate that the synthesized controller stabilizes the real system represented by the uncertain FRF description, the controller is implemented on the wirebonder using a Simulink model. The model is automatically translated into C code and deployed to the FPGA via the Real-Time Interface (RTI) using dSPACE hardware (RTI1200 and MicroLabBox). 

\begin{figure}[t]
	\centering
	\includegraphics[width=0.98\linewidth]{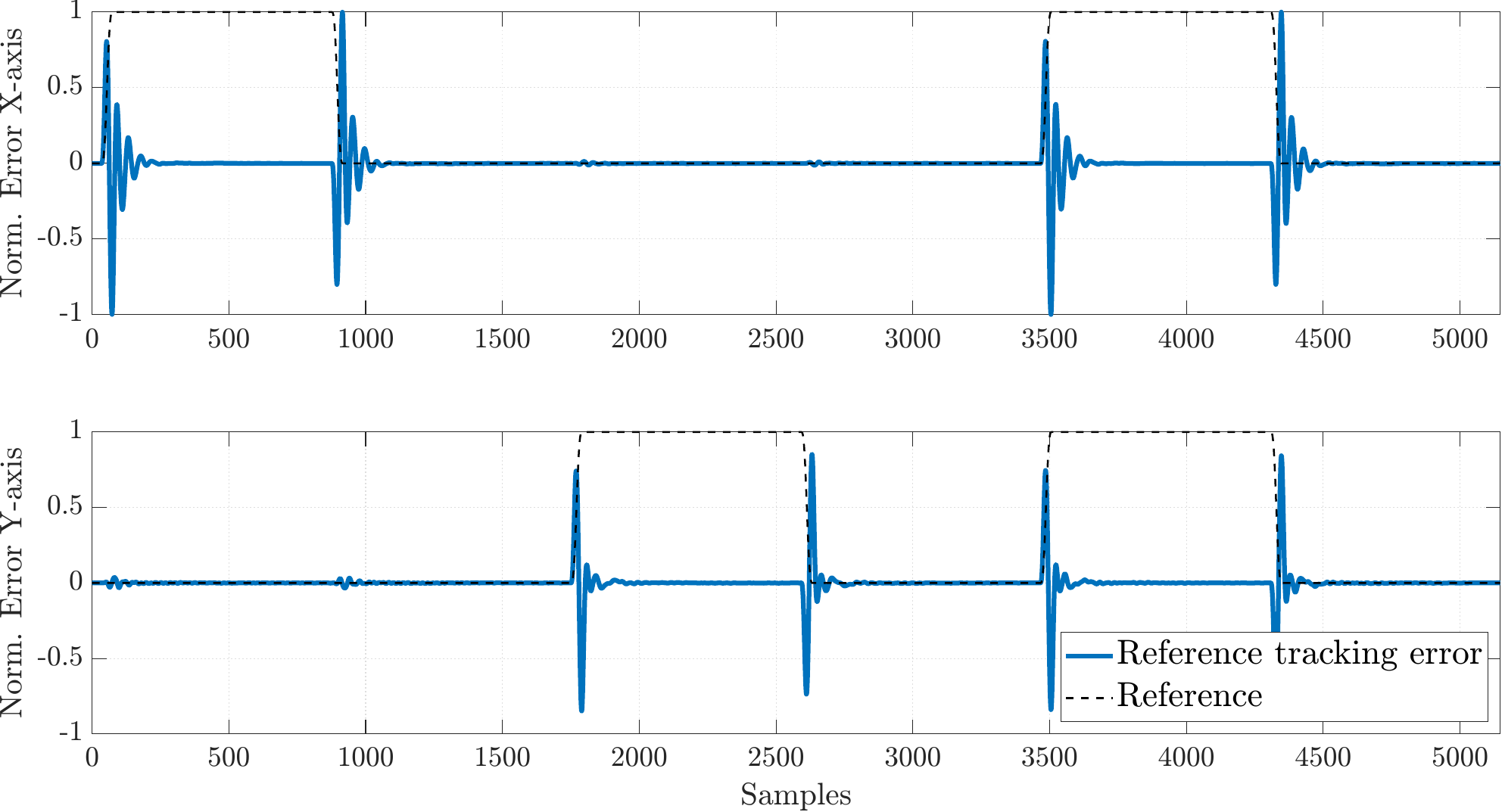}
	\caption{Reference tracking error plot of implemented controller on the isolated motion stage of the wirebonder.}
	\label{fig:TDRobustExp}
\end{figure}

The resulting closed-loop time-domain responses, shown in Figure~\ref{fig:TDRobustExp}, demonstrate stable and well-damped behavior across a range of operating positions, despite significant plant variability and multivariable coupling. These experimental results confirm that the proposed finite-data, frequency-domain synthesis framework yields controllers that not only satisfy the prescribed robustness and performance specifications in analysis, but also translate reliably to a high-precision industrial system without reliance on a parametric model. This provides strong empirical evidence for the practical viability of the proposed approach in realistic, uncertainty-dominated MIMO control applications.


%% file: sections/Conclusion.tex
\section{Conclusions}\label{sec:conclusion}

This paper introduces a novel iteratively convex controller design approach for MIMO systems using only measured, uncertain and sampled frequency response (FRF) data of the plant. The proposed approach explicitly addresses FRF estimation error and spectrum sampling. By transforming infinite-dimensional IQCs into finitely many LMIs, we make the design process computationally efficient while ensuring robust stability and performance guarantees. The method was successfully integrated into an iterative controller design procedure and demonstrated on an industrial MIMO wirebonder, showcasing its practical applicability.

%% file: sections/Acknowledgement.tex
\section*{Acknowledgments}
The authors would like to thank Dragan Kosti\'{c} and Thom Samuels for the invaluable technical discussions and their help in realizing the experimental results. 


%% file: sections/AppendixA.tex
\appendix


\subsection{Convex Sets and Generator Representations} \label{app:convex}

Convex sets are fundamental in robust control, especially when modeling uncertainty in frequency response functions. A set $\mc{S} \subset \mb{K}^{n \times m}$, with $\mb{K} \in \{\mb{R}, \mb{C}\}$, is convex if for any $X, Y \in \mc{S}$ and any $\alpha \in [0,1]$, the convex combination $\alpha X + (1-\alpha) Y$ also belongs to $\mc{S}$. This property ensures that all linear interpolations between points in the set remain inside the set, which is crucial for over-approximating uncertain frequency response data and guaranteeing robust stability.

\begin{definition}
	A finite set $\Gamma = \{\gamma_i\}_{i=1}^{n_\mathrm{g}} \subset \mb{K}^{n \times m}$ of $n_\mathrm{g}$ points is said to \emph{cover} a convex set $\mc{S}$ if every element $M \in \mc{S}$ can be expressed as a convex combination of generators:
	\begin{equation}
		M = \sum_{i=1}^{n_\mathrm{g}} \alpha_i \gamma_i, \quad \alpha_i \ge 0, \quad \sum_{i=1}^{n_\mathrm{g}} \alpha_i = 1.
	\end{equation}
	The convex hull of $\Gamma$, denoted by $\mr{co}(\Gamma)$, is the smallest convex set containing all generators, and, by construction, $\mc{S} \subseteq \mr{co}(\Gamma)$.
\end{definition}

The \emph{center} of the polytope specified by the generators $\Gamma$ is defined as $c(\Gamma) = (n_\mr{g})^{-1} \sum_{i=1}^{n_\mr{g}} \gamma_i$. We say that $\Gamma$ consists of \emph{independent generators} if no generator in the set can be written as a convex combination of the others. This ensures that each generator contributes uniquely to covering the set, which is important for both conceptual clarity and computational efficiency. 
%